\documentclass[11pt]{article}

\title{Deleting and Testing Forbidden Patterns in Multi-Dimensional Arrays} 

\author{
Omri Ben-Eliezer\thanks{Blavatnik School of Computer Science, Tel Aviv University, Tel Aviv, Israel \tt{omrib@mail.tau.ac.il}} \and Simon Korman\thanks{Computer Science Department, University of California at Los-Angeles {\tt{simon.korman@gmail.com}}} \and Daniel Reichman\thanks{Electrical Engineering and Computer Science, University of California at Berkeley {\tt daniel.reichman@gmail.com}}}

\usepackage{multirow}

\newcommand{\Mod}[1]{\ (\text{mod}\ #1)}
\newcommand{\ourParagraph}[1] {\vspace{5pt}\noindent\textbf{\large #1}\vspace{3pt}\\}

\usepackage{enumerate,amssymb,amsmath,amsthm,sectsty,url,wrapfig,fullpage,color}
\usepackage[colorlinks=true]{hyperref}
\usepackage[usenames,dvipsnames]{xcolor}
\usepackage{algorithm}
\usepackage{algorithmic}
\usepackage{graphicx}
\usepackage{bmpsize}

\theoremstyle{plain}
\newtheorem{theorem}{Theorem}

\newtheorem{lemma}[theorem]{Lemma}
\newtheorem{corollary}[theorem]{Corollary}
\newtheorem{claim}[theorem]{Claim}
\newtheorem{observation}[theorem]{Observation}

\theoremstyle{definition}
\newtheorem{definition}[theorem]{Definition}

\theoremstyle{remark}

\newtheorem*{remark*}{Remark}

\begin{document}

\maketitle

\begin{abstract}
Understanding the local behaviour of structured multi-dimensional data is a fundamental problem in various areas of computer science. As the amount of data is often huge, it is desirable to
obtain sublinear time algorithms, and specifically property testers, to understand local properties of the data.

We focus on the natural local problem of testing \emph{pattern freeness}:
given a large $d$-dimensional array $A$ and a fixed $d$-dimensional pattern $P$ over a finite alphabet $\Gamma$,
we say that $A$ is \emph{$P$-free} if it does not contain a copy of the forbidden pattern $P$ as a consecutive subarray.
The distance of $A$ to $P$-freeness is the fraction its entries that need to be modified to make it $P$-free.
For any $\epsilon \in [0,1]$ and any large enough pattern $P$ -- other than a very small set of exceptional patterns --
we design a \emph{tolerant tester} that distinguishes between the case that the distance is at least $\epsilon$
and the case that it is at most $a_d \epsilon$, with query complexity and running time $c_d \epsilon^{-1}$,
where $a_d < 1$ and $c_d$ depend only on the dimension $d$.
For the 1-dimensional case, we provide a linear time algorithm for computing the distance from $P$-freeness.

To analyze the testers we establish several combinatorial results, including the following $d$-dimensional \emph{modification lemma}, which might be of independent interest: for any large enough $d$-dimensional pattern $P$ over any alphabet (excluding a small set of exceptional patterns for the binary case), and any $d$-dimensional array $A$ containing a copy of $P$, one can delete this copy by modifying one of its locations without creating new $P$-copies in $A$.

Our results address an open question of Fischer and Newman, who asked whether there exist efficient testers for properties related to tight substructures in multi-dimensional structured data. They serve as a first step towards a general understanding of local properties of multi-dimensional arrays.
\end{abstract} 


\section{Introduction}
Pattern matching is the algorithmic problem of finding occurrences of a fixed pattern in a given string. This problem appears in many settings and has applications in diverse domains such as computational biology, computer vision, natural language processing and web search. There has been extensive research concerned with developing algorithms that search for patterns in strings, resulting with a wide range of efficient algorithms \cite{Boyer,knuth,Galil,Czumaj,Navarro,Lecroq}. Higher-dimensional analogues where one searches for a $d$-dimensional pattern in a $d$-dimensional array have received attention as well. For example, the 2D case arises in analyzing aerial photographs \cite{Amir0,Amir1} and the 3D case has applications in medical imaging.

Given a string $S$ of length $n$ and a pattern $P$ of length $k\le n$, any algorithm which determines whether $P$ occurs in $S$ has running time $\Omega(n)$ \cite{Cole,Rivest} and a linear lower bound carries over to higher dimensions. For the 2D and 3D case, when an $n \times n$ image is concerned, algorithms whose run time is $O(n^2)$ are known \cite{Amir1}. These algorithms have been generalized to the 3D case to yield running time of $O(n^3)$ \cite{Galil3D}. Finally it is also known (e.g., \cite{Kark}) that for the $d$-dimensional case it is possible to solve the pattern matching problem in time $O(d^2 n^d\log m )$ (where the pattern is an array of size $m^d$). It is natural to ask which tasks of this type can be performed in sublinear (namely $o(n^d)$) time for $d$-dimensional arrays.

The field of \emph{property testing} \cite{Goldreich,Rubinfeld} deals with decision problems regarding discrete objects (e.g., graphs, functions, images) that either have a certain property $P$ or are \emph{far} from satisfying $P$. Here, we are interested in deciding quickly whether a given $d$-dimensional array $A$ is far from \emph{not} containing a fixed $d$-dimensional pattern $P$.
\emph{Tolerant} property testing \cite{Parnas} is a useful extension of the standard notion, in which the tester needs to distinguish between objects that are \emph{close} to satisfying the property to those that are \emph{far} from satisfying it. 


A $d$-dimensional $k_1 \times \ldots \times k_d$ array $A$ over
an alphabet $\Gamma$ is a function from $[k_1] \times \ldots \times [k_d]$ to $\Gamma$.
For simplicity of presentation, all results in this paper will be presented for cubic arrays in which $k_1 = \ldots = k_d$,
but they generalize to non-cubic arrays in a straightforward manner.
We consider the (tolerant) \emph{pattern-freeness problem} where one needs to distinguish between the case that a given $d$-dimensional array $A$ is
$\epsilon_1$-close to being $P$-free for a fixed pattern $P$, and the case that $A$ is $\epsilon_2$-far from being $P$-free, where $\epsilon_1 < \epsilon_2$.
An $(\epsilon_1, \epsilon_2)$-tester $Q$ for this problem is a randomized algorithm that is given access to an array $A$, as well as its size and proximity parameters $0 \leq \epsilon_1 < \epsilon_2 < 1$. $Q$ needs to distinguish with probability at least $2/3$ between the case that $A$ is is $\epsilon_1$-close to being $P$-free and the case that $A$ is $\epsilon_2$-far from being $P$-free. The query complexity of $Q$ is the number of queries it makes in $A$. 

Our interest in the pattern-freeness problem stems from several applications. In certain scenarios of interest, we might be interested in identifying quickly that an array is far from not containing a given pattern. For the one dimensional case, being far from not containing a given text may indicate a potential anomaly which requires attention (e.g., an offensive word in social network media), hence such testing algorithms may provide useful in anomaly detection. Many computer vision methods for classifying images are feature based: hence being far from containing a certain pattern associated with a feature may be useful in rejection methods that enable to quickly discard images that do not possess a certain visual property.

Beyond practical applications, devising property testing algorithms for the pattern freeness problem is of theoretical interest. In the first place, it leads to a combinatorial characterization of the distance from being $P$-free. Such a characterization has proved fruitful in graph property testing \cite{Alon,Large} where celebrated graph removal lemmas were developed en route of
devising algorithms for testing subgraph freeness. We encounter a similar phenomena in studying patterns and arrays: at the core of our approach for testing pattern freeness lies a \emph{modification lemma} for patterns which we state next. We believe that this Lemma may be of independent interest and find applications beyond testing algorithms. Later we show one such application: computing the exact distance of a (one dimensional) string from being $P$-free can be done in linear time.

For a pattern $P$ of size $k\times k\times \ldots \times k$, any of its entries that is in $\{0,k-1\} \times \ldots \times \{0, k-1\}$ is said to be a \emph{corner} of $P$.
We say that $P$ is \emph{almost homogeneous} if all of its entries but one are equal, and the different entry lies in a corner of $P$. Finally, $P$ is \emph{removable} (with respect to the
alphabet $\Gamma$) if for any $d$-dimensional array $A$ over $\Gamma$ and any copy of $P$ in $A$, one can
destroy the copy by modifying one of its entries without creating new $P$-copies in $A$.
The modification lemma states that for any $d$, and any large enough pattern $P$, when the alphabet is binary it holds that $P$ is removable \emph{if and only if} it is not almost homogeneous, and when the alphabet is not binary, $P$ is removable provided that it is large enough.



Recent works \cite{Berman,Berman2} have obtained tolerant testers for visual properties. As observed in \cite{Berman,Berman2}, tolerance is an attractive property for testing visual properties as real-world images are often noisy.
With the modification Lemma at hand, we show that when $P$ is removable, the (relative) \emph{hitting number} of $P$ in $A$, which is the minimal size of a set of entries that intersects all $P$-copies in $A$ divided by $|A|$, differs from the distance of $A$ from $P$-freeness by a multiplicative factor that depends only on $d$ (and not on $P$ or $A$). 
This relation allows us to devise very fast $(5^{-d}\epsilon ,\epsilon)$-tolerant testers for $P$-freeness, as the hitting number of $P$ in $A$ can be well approximated using only a very small sample of blocks of entries from $A$.
The query complexity of our tester is $O(C_d/\epsilon)$, where $C_d$ is a positive constant depending only on the dimension $d$ of the array. Note that our characterization in terms of the hitting number is crucial: merely building on the fact that $A$ contains many occurrences of $P$ (as can be derived directly from the modification lemma) and randomly sampling $O(1/\epsilon)$ possible locations in $A$, checking whether the sub-array starting at these locations equals $P$ would lead to query complexity of $O(k^d/\epsilon)$. Note that our tester is optimal (up to a multiplicative factor that depends on $d$), as any tester for this problem makes $\Omega(1/\epsilon)$ queries.

The one dimensional setting, where one seeks to determine quickly whether a string $S$ is $\epsilon$-far from being $P$-free is of particular interest. We are able to leverage the modification Lemma and show that the distance of a string $S$ from being $P$-free for a fixed pattern $P$ (that is not almost homogeneous) is \emph{exactly equal} to the hitting number of $P$ in $A$. For an arbitrary constant $0<c<1$, this characterization allows us to devise a $((1-c)\epsilon, \epsilon)$-tolerant tester making $O_c(\epsilon^{-1})$ queries for this case.
For the case of almost homogeneous patterns, and an arbitrary constant $c>0$ , we devise a $((1/16+c)\epsilon, \epsilon)$-tolerant tester that makes $O_c(1/\epsilon)$ queries.
Whether tolerant testers exist for almost homogenous patterns of dimension larger than $1$ is an open question.

Moreover, the characterization via the hitting number implies an $O(n+k)$ algorithm that calculates (exactly) the distance of $A$ from being $P$-free where $P$ is an \emph{arbitrary} pattern (that may be almost homogeneous). We are not aware of a previous algorithm for the distance computation problem.
Unlike the one-dimensional case, in $d$ dimensions we do not know of a clean combinatorial description of the distance to being $P$-free for higher dimension. Furthermore, it can be shown via a direct reduction from covering problems in the plane \cite{Fowler}, that for dimension $d>1$ there exists patterns $P$ for which calculating the distance to $P$-freeness is NP-hard.
\section{Related Work}
The problem of testing pattern freeness is related to the study of testing subgraph-freeness (see, for example, \cite{Bipartite,Large}). This line of work examines how one can test quickly whether a given graph $G$ is $H$-free or $\epsilon$-far from being $H$-free, where $H$ is a fixed subgraph. In this problem, a graph is $\epsilon$-far from being $H$-free if at least an $\epsilon$-fraction of its edges and non-edges need to be altered in order to ensure that the resulting graph does not contain $H$ as a (not necessarily induced) subgraph. A key component in these works are removal lemmas: typically such lemmas imply that if $G$ is $\epsilon$-far from being $H$-free, it necessarily contains a ``large" number of copies of $H$. Perhaps the best example for this phenomena is the triangle removal lemma which asserts that for every $\epsilon \in (0,1)$, there exists $\delta=\delta(\epsilon)>0$ such that if an $n$-vertex graph $G$ is $\epsilon$-far from being triangle free, then $G$ contains at least $\delta n^3$ triangles (see e.g., \cite{AS} and the reference within).

Alon \emph{et. al.} showed \cite{Alon} that regular languages over $\{0,1\}$ are strongly testable. Testing pattern-freeness (1-dimensional, binary alphabet, constant pattern length $k$) is a special case of the former, since the language of all strings avoiding a fixed pattern is regular. The query complexity of their tester is $O\left(\frac{c}{\epsilon} \cdot \ln^3(\frac{1}{\epsilon})\right)$, where $c$ is a constant that depends on the minimal size of a DFA $A_L$, that accepts the regular language $L$. It is shown in \cite{Alon} that $c$ can be taken to be $O(s^3)$ where $s$ is the size of $A_L$. In the case of the regular language considered here a simple pumping-lemma inspired argument shows that $s \geq \Omega(k)$. 
Hence the upper bound on testing pattern freeness implied by their algorithm is $O\left(\frac{ k^3}{\epsilon} \cdot \ln^3(\frac{1}{\epsilon})\right)$.
Our 1D tester solves a very restricted case of the problem the tester of \cite{Alon} deals with, but it achieves a better query complexity of $O(1/\epsilon)$ in this setting.
Moreover, our tester is much simpler and can be applied in the more general high dimensional setting, or when the pattern length $k$ is allowed to grow as a function of the string length $n$.

The problem of testing \emph{submatrix freeness} was investigated in \cite{NewmanStoc,Newman,Induced,Rozenberg,ABen}. As opposed to our case, which is concerned with \emph{tight} submatrices, all of these results deal with submatrices that are not necessarily tight (i.e.\@ the rows and the columns need not be consecutive). Quantitatively, the submatrix case is very different from our case: in our case $P$-freeness can be testable using $O(\epsilon^{-1})$ queries, while in the submatrix case, for a binary submatrix of size $k \times k$ a lower bound of $\epsilon^{-\Omega(k^2)}$ on the needed number of queries is easy to obtain, and in the non-binary case there exist $2 \times 2$ matrices for which there exists a super polynomial lower bound of $\epsilon^{\Omega(\log 1/\epsilon)}$.

The 2D part of our work adds to a growing literature concerned with testing properties of images \cite{sofya,Ron,Berman}.
Ideas and techniques from the property testing literature have recently been used in the fields of computer vision and pattern recognition \cite{Kleiner,Korman}.



\section{Notation and definitions}\label{sec:def}
With slight abuse of notation, for a positive integer $n$ we let $[n]$ denote the set $\{0,\ldots,n-1\}$ and we write
$[n]^d = [n] \times \ldots \times [n]$.

Recall that a $d$-dimensional (cubic) array $A$ over
an alphabet $\Gamma$ is a function from $[k]^d$ to $\Gamma$.
The $x = (x_1, \ldots, x_d)$ \emph{entry} of $A$, denoted by $A_x$, is the value of the function $A$ at location $x$.
Let $P$ be a  $(k,d)$-array over an alphabet $\Gamma$ of size at least two.
We say that a $d$-dimensional array $A$ \emph{contains} a copy of
$P$ (or a $P$-copy) \emph{starting in location} $x = (x_1, \ldots, x_d)$ if for any $y \in [k]^d$ we have
$A_{x + y} = P_{y}$. Finally, $A$ is $P$-free if it does not contain copies of $P$.


A property $\mathcal{P}$ of $d$-dimensional arrays is simply a family of such arrays over an alphabet $\Gamma$.
For an array $A$ and a property $\mathcal{P}$, the \emph{absolute distance} $d_{\mathcal{P}}(A)$ of $A$ to $\mathcal{P}$ is the minimal number of entries that one needs to change in $A$ to get an array from $\mathcal{P}$. The \emph{relative distance} of $A$ to $\mathcal{P}$ is $\delta_{\mathcal{P}}(A) = d_{\mathcal{P}}(A) / |A|$, where clearly $0 \leq \delta_\mathcal{P}(A) \leq 1$ for any nontrivial $\mathcal{P}$ and $A$.
We say that $A$ is $\epsilon$-close ($\epsilon$-far) to $\mathcal{P}$ if $\delta_{\mathcal{P}}(A) \leq \epsilon$ ($\geq \epsilon$).

In this paper we will consider the property of $P$-freeness, which consists of all $P$-free arrays.
The absolute and relative distance to $P$-freeness will be denoted by $d_P(A)$ and $\delta_P(A)$, respectively.

For an array $A$ and a pattern $P$ we will call a set of entries in $A$ whose modification can turn it to be $P$-free a \emph{deletion set} and therefore it is natural to call $d_P(A)$ (the absolute distance of $A$ to $P$-freeness) the \emph{deletion number}, since it is the size of a minimal deletion set. In a similar manner, for a given set of entries in $A$, if every $P$-copy in $A$ contains at least one of these entries, we call it a \emph{hitting set} and we call the size of a minimal hitting set the \emph{hitting number}, denoted by $h_P(A)$.
For all notations here and above, in the 1-dimensional case we will replace $A$ by $S$ (for String).

We will need several definitions from \cite{Parnas} as well.
Let $\mathcal{P}$ be a property of arrays and let $h_1, h_2: [0,1] \to [0,1]$ be two monotone increasing functions.
An $(h_1, h_2)$-\emph{distance approximation} algorithm for $\mathcal{P}$ is given query access to an unknown array $A$.
The algorithm outputs an estimate $\hat{\delta}$ to $\delta_P(A)$, such that with probability at least $2/3$ it holds that
$h_1(\delta_P(A)) \leq \hat{\delta} \leq h_2(\delta_P(A))$.
Finally, for a property $\mathcal{P}$ and for $0 \leq \epsilon_1 < \epsilon_2 \leq 1$, an $(\epsilon_1, \epsilon_2)$-\emph{tolerant tester}
for $\mathcal{P}$ is given query access to an array $A$. The tester \emph{accepts} with probability at least $2/3$ if
$A$ is $\epsilon_1$-close to $\mathcal{P}$, and \emph{rejects} with probability at least $2/3$ if $A$ is $\epsilon_2$-far from $\mathcal{P}$.
In the `standard' notion of property testing, $\epsilon_1 = 0$. Thus, any tolerant tester is also a tester in the standard notion.
Finally, we define the additive (multiplicative) \emph{tolerance} of the tester above as $\epsilon_2 - \epsilon_1$ ($\epsilon_2 / \epsilon_1$ respectively).


\section{Main Results}


The modification lemma result is central in the study of minimal deletion sets. It classifies the possible patterns into ones that are removable and ones that are not. The result that the vast majority of patterns are removable is used extensively throughout the paper in the design and proofs of algorithms for efficient testing of pattern freeness (in 1 and higher dimensions) as well as for the exact computation of the deletion number in 1-dimension.

Our 1-dimensional modification lemma (Lemma~\ref{thm:modification_1D}) gives the following full characterization of 1-dimensional patterns (i.e. strings). A binary pattern is removable \emph{if and only if} it not almost homogeneous, while \emph{any} pattern over a larger alphabet is removable. The multidimensional version of the lemma (Lemma~\ref{thm:modification_1D}) makes the exact same classification, but for $(k,d)$-arrays for which $k\ge3\cdot 2^d$.

The fact that most patterns are removable is very important for analyzing the deletion number (which is the distance to pattern freeness). As an example, a simple observation is that a removable pattern appears at least $d_P(A)$ times (possibly with overlaps) in the array $A$, which implies an $\epsilon$-tester that can simply check for the presence of the pattern in $1/\epsilon$ random locations in the array at a sample complexity of $O(k/\epsilon)$.

Another important part of our work makes explicit connections between the deletion number and the hitting number for both 1 and higher dimensions. These are needed in order to get improved testers (e.g. for getting rid of $k$ in the sample complexity) in $d$-dimensions as well as for linear time computation of the distance (deletion number) in 1-dimension.

For the 1-dimensional case we show that the deletion number $d_P(S)$ equals the hitting number $h_P(S)$, which leads to an exact computation of $d_P(S)$ in time $O(n+k)$ (Theorem~\ref{thm:approx1D}) as well as a tolerant testers for Pattern Freeness: An $(\epsilon_1, \epsilon_2)$-tolerant tester for any $0 \leq \epsilon_1 < \epsilon_2 \leq 1$ at a complexity of $O(\epsilon_2^2 / (\epsilon_2 - \epsilon_1)^3)$ (Theorem~\ref{thm:test_1D}) as well as an
$((1-\tau)\epsilon, \epsilon)$-tolerant tester for a fixed $\tau > 0$ and any $0 < \epsilon \leq 1$ at a complexity of $O(\epsilon^{-1} \tau^{-3})$ (Corollary~\ref{cor:test_1D.multiplicative}).

For higher dimensions, we show (Lemma~\ref{lem:hitting_set_2D}) that $h_P(A) \leq d_P(A) \leq \alpha_d h_P(A) \leq \alpha_d k^{-d}$, a bound that relates the hitting number $h_P(A)$ and the deletion number $d_P(A)$ through a constant $\alpha_d = 4^d + 2^d$ that depends only on the dimension $d$. This bound enables a $((1-\tau)^d \alpha_d^{-1}\epsilon, \epsilon)$-tolerant tester making $C_\tau \epsilon^{-1}$ queries, where $C_\tau = O(1 / \tau^d (1 - (1-\tau)^d)^2)$ (Theorem~\ref{thm:test_highD}).


In the 1-dimensional setting we also provide dedicated algorithms to handle the almost homogeneous (non-removable) patterns, achieving an $O(n)$ algorithm for computing the deletion number (Theorem~\ref{thm:compute_almost_homo}) as well as a $(\epsilon/(16+c), \epsilon)$-tolerant tester, for any constant $c > 0$, at a complexity of $\alpha_c \epsilon^{-1}$ queries, where $\alpha_c$ depends only on $c$ (Theorem~\ref{thm:test_almost_homo}).

Finally, we provide a lower bound of $\Omega(1/\epsilon)$ (Theorem~\ref{thm:LB}, Appendix~\ref{sec.lower.bound}) for any tester of pattern freeness. Unlike the previous lower bound of $\Omega(1/\epsilon)$ \cite{Alon} on testing regular languages, ours extends to dimensions higher than 1 and applies to the case where $k$ may depend on $n$.

Our main results are summarized in Table~\ref{table.results}.


\begin{table}[h!] \vspace{-0pt}
  \centering
  \addtolength{\tabcolsep}{-1pt}{\small
  \begin{tabular}{|c|l|cccc|}
    \hline
    \multirow{2}{*}{dim.} & \multirow{2}{*}{template type}
    &\small{deletion number}&\multirow{2}{*}{modification lemma}&\small{tester}&\small{query} \\
    & &\small{computation}&&\small{tolerance}&\small{complexity} \\
    \hline
    \hline
    \multirow{2}{*}{\small{1D}}
    & \small{{general}}
    & \small{$O(n+k)$} & \small{removable for any $k$} & $1/(1-\tau)$ & $O(1/\epsilon\tau^{3})$  \\
    & \small{{almost homog.}}
    & $O(n+k)$ & not removable for any $k$ & $(16+c)$ & $\alpha_c/ \epsilon$   \\
    \hline
    \multirow{2}{*}{\small{$2^+$D}}
    & \small{{general}}
    & NP-Hard & \small{removable for $k>3\cdot2^d$} & $(1-\tau)^{-d} \alpha_{d}$ & $\beta_{d,\tau}/\epsilon$  \\
    & \small{{almost homog.}}
    & $-$ & not removable for any $k$ & $-$ & $-$     \\
    \hline
  \end{tabular}}\vspace{2pt}
  \caption{\textbf{Summary of results.} $0 < \tau < 1$ and $c > 0$ are arbitrary constants. $\alpha_c$ is a constant that depends only on $c$. $\beta_{d,\tau}$ is a constant that depends only on $d$ and $\tau$. 'modification lemma' specifies if patterns are classified as removable or not. the 'tester tolerance' is multiplicative} \vspace{-9pt}
  \label{table.results}
\end{table}

\section{Modification Lemma}\label{sec.modification.lemmas}


\begin{theorem}[Modification Lemma]\label{thm:modification}
Let $d > 1$ and let $P$ be a $(k,d)$-array over the alphabet $\Gamma$ where $k \geq 3 \cdot 2^{d}$.
\begin{enumerate}
\item If $|\Gamma|=2$ then $P$ is removable if and only if it is not almost homogeneous.
\item If $|\Gamma|\geq 3$ then $P$ is removable.
\end{enumerate}
\end{theorem}

\begin{remark*}
Theorem \ref{thm:modification} states that any large enough binary pattern which is not almost homogeneous is removable.
The requirement that the pattern is large enough is crucial, as the $2 \times \ldots \times 2$ pattern $P$ satisfying $P_x = 0$
for any $x = (x_1, \ldots, x_d)$ with $x_1 = 0$ and $P_x = 1$ otherwise is not removable even though it is not almost homogeneous.
To see this, consider the following $4 \times \ldots \times 4$ array $A$:
$M_x = 0$ if either $x_1 = 0$, or $x_1 = 1$ and $x_i \in \{1,2\}$ for any $2 \leq i \leq d$, or $x_1 = 2$ and $x_i \in \{0,3\}$ for some $2 \leq i \leq d$. For any other value of $x$, $M_x = 1$.
Note that $A$ contains a copy of $P$ starting at $(1, \ldots, 1)$, but flipping any bit in this copy creates a new $P$-copy in $A$.
Still, the size of the counterexample is only $2 \times \ldots \times 2$ while in the statement of Theorem \ref{thm:modification}, the dependence is exponential in $d$. It will be interesting to understand what is the correct order of magnitude of the dependence of $k$ on $d$.
\end{remark*}

\begin{proof}[Proof of Theorem \ref{thm:modification}]
The second statement of the theorem can be easily derived from the first statement;
If $P$ does not contain all letters in $\Gamma$ then it is clearly removable, as changing any of its entries to any of the missing letters cannot create new $P$-copies. Otherwise, we can reduce the problem to the binary case: let $\sigma_1, \sigma_2$ be the letters in $\Gamma$ that appear the smallest number of times in $P$. Consider the following $(k,d)$-array $P'$ over $\{0,1\}$: $P'_x = 0$ if $P_x \in \{\sigma_1, \sigma_2\}$
and $P'_x = 1$ otherwise. Observe that $P'$ is not almost homogeneous, implying that it is removable. It is not hard to verify now that $P$ is removable as well.

In what follows, we will prove the first statement. If $P$ is binary and almost homogeneous then it
is not removable: Without loss of generality $P_{(0,\ldots,0)}=1$ and $P_x = 0$ for any $x \neq (0, \ldots, 0)$.
Consider a $(2k,d)$-array $A$ such that $M_{(0, \ldots, 0)} = M_{(1, \ldots, 1)} = 1$ and $A=0$ elsewhere. Clearly, modifying any bit of the $P$-copy starting at $(1, \ldots, 1)$ creates a new copy of $P$ in $A$, so $P$ is not removable.

The rest of the proof is dedicated to the other direction.
Suppose that $P$ is a binary $(k,d)$-array that is not removable. We would like to show that $P$ must be almost homogeneous.
As $P$ is not removable, there exists a binary array $A$ containing a copy of $P$ that such that flipping
any single bit in this copy creates a new copy of $P$ in $A$. This copy of $P$ will be called the \emph{template} of $P$ in $A$.

Clearly, all of the new copies created by flipping bits in the template must intersect the template, so we may assume that $A$ is of size $(3k-2)^d$ and that the template starts in location $k = (k-1,\ldots, k-1)$.

For convenience, let $I = [k]^d$ denote the set of indices of $P$.
For any $x \in I$ let $\bar{x} = x + k$; $\bar{x}$ is the location in $A$ of bit $x$ of the template.

Roughly speaking, our general strategy for the proof would be show that there exists at most two "special" entries in $P$ such that when we flip a bit in the template, creating a new copy of $P$ in $A$, the flipped bit usually plays the role of one of the special entries in the new copy.
We will then show that in fact, there must be exactly one special entry, which must lie in a corner of $P$, and that all non-special entries are equal while the special entry is equal to their negation. This will finish the proof that $P$ is almost homogeneous.

\begin{definition}
Let $i \leq d$ and let $\delta$ be positive integers. Let $x = (x_1, \ldots, x_d)$ and $y = (y_1, \ldots, y_d)$ be $d$-dimensional points.
The pair $(x,y)$ is \emph{$(i, \delta)$-related} if $y_i - x_i = \delta$ and $y_j = x_j$ for any $j \neq i$.
An $(i, \delta)$-related pair $(x,y)$ is said to be an \emph{$(i, \delta)$-jump} in $P$ if $P_x \neq P_y$.
\end{definition}

\begin{figure} [h!] \vspace{1pt}
\begin{center}
\small\addtolength{\tabcolsep}{3pt}
\hspace{-52pt}
\begin{tabular}{p{0.41\textwidth}p{0.4\textwidth}}
\vspace{-180pt}
\caption{\textbf{Illustration for Lemma~\ref{lem:related_pairs}}. A 2-dimensional example, where $i$ is the vertical coordinate: Flipping the bit (of the template $P$) at location $\bar{a}$ creates the $P$-copy $Q^a$ at location $m(a)$. Similarly, the copy $Q^b$ is created at location $m(b)$. Note that the pair of points $(\bar{x},\bar{y})$ (which is $(x,y)$ in $P$) and the copy locations pair $(m(a),m(b))$ are both $(i, \Delta_i)$-related. The values $P_x$ and $P_y$ ($M_{\bar{x}}$ and $M_{\bar{y}}$) must be equal.}
&
\includegraphics[width = 0.53\textwidth]{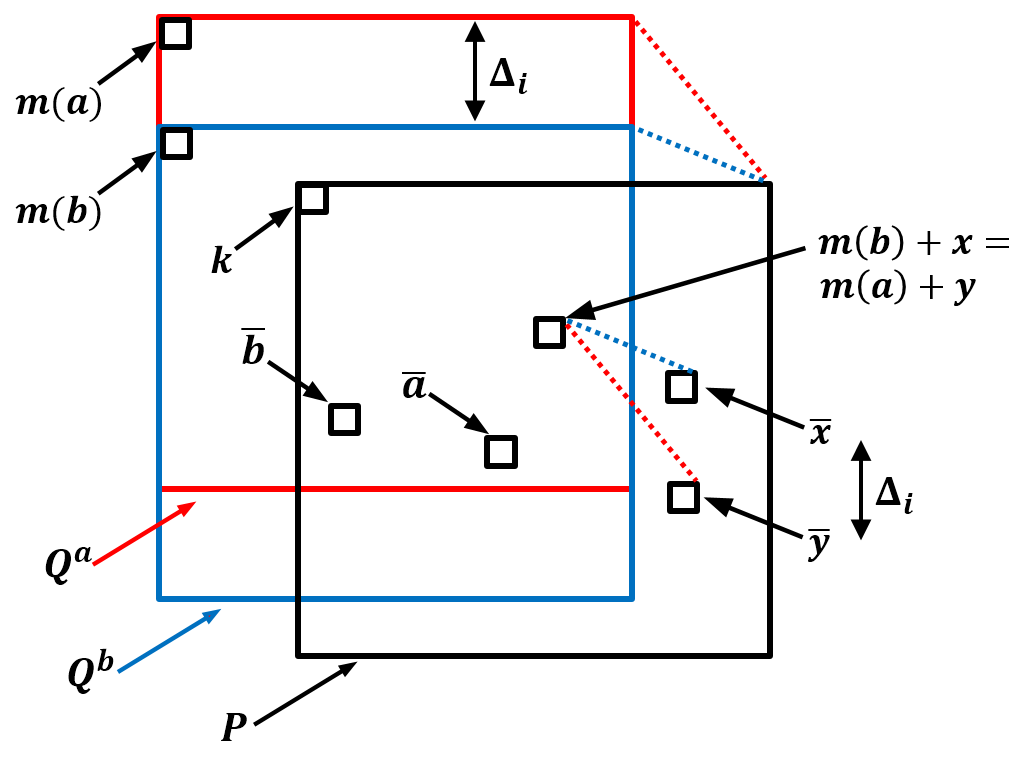}
\end{tabular}
\end{center}
\vspace{-19pt}
\end{figure}

\begin{lemma} \label{lem:related_pairs}
For any $1 \leq i \leq d$ there exists $0 < \Delta_i < k/3$ such that at most two of the $(i, \Delta_i)$-related pairs of points
from $I$ are $(i, \Delta)$-jumps in $P$.
\end{lemma}
\begin{proof}
Recall that, by our assumption, flipping any of the $K = k^d$ bits of the template creates a new copy of $P$ in $A$.
Consider the following mapping $m:I \to [2k-1]^d$.
$m(x_1, \ldots, x_d)$ is the starting location of
a new copy of $P$ created in $A$ as a result of flipping bit $x = (x_1, \ldots, x_d)$ of the template (which is bit $\bar{x}$ of $A$).
If more than one copy is created by this flip, then we choose the starting location of one of the copies arbitrarily.

Observe that $m$ is injective, and let $S$ be the image of $m$, where $|S| = K$.
Let $1 \leq i \leq d$ and consider the collection of (one-dimensional) lines $$\mathcal{L}_i = \big\{ \{x_1\} \times \ldots \times \{x_{i-1}\} \times [2k - 1] \times \{x_{i+1}\} \times \ldots \times \{x_d\} \quad | \quad \forall j \neq i: x_j \in [2k-1] \big\}.$$
Clearly $\sum_{\ell \in \mathcal{L}_i} |S \cap \ell| = K$. On the other hand,
$|\mathcal{L}_i| = \prod_{j \neq i} (2k-1) < 2^{d-1} \prod_{j \neq i} k = 2^{d-1} K / k$,
so there exists a line $\ell \in \mathcal{L}_i$ for which $|S \cap \ell| > k / 2^{d-1} \geq 6$. Hence $|S \cap \ell| \geq 7$.
Let $\alpha_1 < \ldots < \alpha_7$ be the smallest $i$-indices of elements in $S \cap \ell$. Since $\alpha_7 - \alpha_1 < 2k-1$
there exists some $1 \leq l \leq 6$ such that $\alpha_{l+1} - \alpha_l < k / 3$. That is, $S$ contains an $(i,\Delta_i)$-related pair
with $0 < \Delta_i < k / 3$. In other words, there are two points $a,b \in I$ such that flipping $\bar{a}$ ($\bar{b}$)
would create a new $P$-copy, denoted by $Q^a$ ($Q^b$ respectively), which starts in location $m(a)$ ($m(b)$ respectively) in $A$, and $(m(a), m(b))$ is an $(i, \Delta_i)$-related pair.

The following claim finishes the proof of the lemma and will also be useful later on.
\begin{claim}
\label{clm:jumps}
For $a$ and $b$ as above, let $(x,y)$ be a pair of points from $I$ that are $(i, \Delta_i)$-related and suppose that $y \neq \bar{a} - m(a)$ and that $x \neq \bar{b} - m(b)$. Then $P_x = P_y$.
\end{claim}
\begin{proof}
The bits that were flipped in $A$ to create $Q^a$ and $Q^b$ are $\bar{a}, \bar{b}$ respectively.
Since $y + m(a) \neq \bar{a}$, the copy $Q_a$ contains the original entry of $A$ in location $y + m(a)$. Therefore, $P_y = M_{y+m(a)}$ (as $M_{y+ m(a)}$ is bit $y$ of $Q^a$, which is a copy of $P$). Similarly, since $x + m(b) \neq \bar{b}$, we have $P_x = M_{x+m(b)}$.
But since both pairs $(x,y)$ and $(m(a), m(b))$ are $(i, \Delta_i)$-related, we get that
$m(b) - m(a) = y - x$, implying that
 $x + m(b) = y + m(a)$, and therefore
$P_x = M_{x+m(b)} = M_{y+m(a)} = P_y$, as desired.
\end{proof}

Clearly, the number of $(i, \Delta_i)$-related pairs that do not satisfy the conditions of the claim is at most two, finishing the proof of Lemma \ref{lem:related_pairs}.
\end{proof}
Let $\Delta = (\Delta_1, \ldots,\Delta_d)$ where for any $1 \leq i \leq d$, we take $\Delta_i$ that satisfies the statement of Lemma \ref{lem:related_pairs} (its specific value will be determined later).
\begin{definition}
Let $x \in I$.
The set of  \emph{$\Delta$-neighbours} of $x$ is
\vspace{-6pt}\[
N_x = \left\{ y \in I \quad \big| \quad \exists i: (x,y) \text{ is $(i, \Delta_i)$-related or } (y,x) \text{ is $(i, \Delta_i)$-related}\right\}
\vspace{-4pt}\]
and the number of $\Delta$-neighbours of $x$ is $n_x = |N_x|$, where $d \leq n_x \leq 2d$.
We say that $x$ is a $\Delta$-\emph{corner} if $n_x(\Delta) = d$ and that it is $\Delta$-\emph{internal} if $n_x(\Delta) = 2d$. Furthermore,
$x$ is \emph{$(\Delta, P)$-isolated} if $P_x \neq P_y$ for any $y \in N_x$, while it is \emph{$(\Delta, P)$-generic} if $P_x = P_y$ for any $y \in N_x$.
\end{definition}
When using the above notation, we will sometimes omit the parameters (e.g.\@ simply writing \emph{isolated} instead of $(\Delta, P)$-isloated) as the context is usually clear.

The definition imposes a symmetric neighborhood relation, that is, $x \in N_y$ holds if and only if $y \in N_x$. If $x \in N_y$ we say that $x$ and $y$ are $\Delta$-neighbours.
Note that a point $x = (x_1, \ldots, x_d) \in I$ is a $\Delta$-corner if $x_i < \Delta_i$ or $x_i \geq k - \Delta_i$ for any $1 \leq i \leq d$, and that $x$ is $\Delta$-internal if $\Delta_i \leq x_i < k - \Delta_i$ for any $1 \leq i \leq d$.
\begin{claim}\label{clm:no_isolated_neighbors}
Two $(\Delta, P)$-isolated points in $I$ cannot be $\Delta$-neighbors.
\end{claim}
\begin{proof}
Suppose towards contradiction that $x = (x_1, \ldots, x_d)$ and $y = (y_1, \ldots, y_d)$ are two distinct
$(\Delta, P)$-isolated points and that $(x,y)$ is $(i,\Delta_i)$-related for some $1 \leq i \leq d$. Since $\Delta_i < k/3$, at least one of $x$ or $y$ participates in two different $(i,\Delta_i)$-related pairs: if $x_i < k / 3$ then $y_i + \Delta_i = x_i + 2\Delta_i < k$ so $y$ is in two such pairs, and otherwise $x_i \geq \Delta_i$, meaning that $x$ participates in two such pairs. Assume without loss of generality that the two $(i, \Delta_i)$-related pairs are $(t,x)$ and $(x,y)$, then $P_t\neq P_x$ and $P_x\neq P_y$ as $x$ is isolated.
By Lemma~\ref{lem:related_pairs}, these are the only $(i,\Delta_i)$-jumps in $P$.

Choose an arbitrary $j \neq i$ and take $v = (v_1, \ldots, v_d)$ where $v_j = \Delta_j$ and $v_{l} = 0$ for any $l \neq j$. Recall that $\Delta_j < k / 3$, implying that either $x_j + v_j < k$ or $x_j - v_j \geq 0$. Without loss of generality assume the former, and let $x' = x + v$ and $y' = y + v$.
Since $x$ and $y$ are $(\Delta, P)$-isolated, and since $x' \in N_x$ and $y' \in N_y$, we get that
$P_{x'} \neq P_x \neq P_y \neq P_{y'}$, and thus $P_{x'} \neq P_{y'}$ (as the alphabet is binary). Therefore, $(x', y')$ is also an $(i, \Delta_i)$-jump in $P$, a contradiction.
\end{proof}
\vspace{3pt}
\hspace{8pt}\begin{tabular}{p{0.4\textwidth}p{0.52\textwidth}}\vspace{5pt}
{\textbf{Illustration for Definition~\ref{def.mapped.to}}. Recall that flipping a bit $\bar{a}$ in $A$ creates a new $P$-copy $Q^a$ (which contains $\bar{a}$), located at the point $m(a)$ in the coordinates of $A$. The bits $x$ and $a$ are \emph{mapped} to $y$ and $f(a)$ respectively.}
&
  \begin{center}\vspace{-21pt}
    \includegraphics[width=0.38\textwidth]{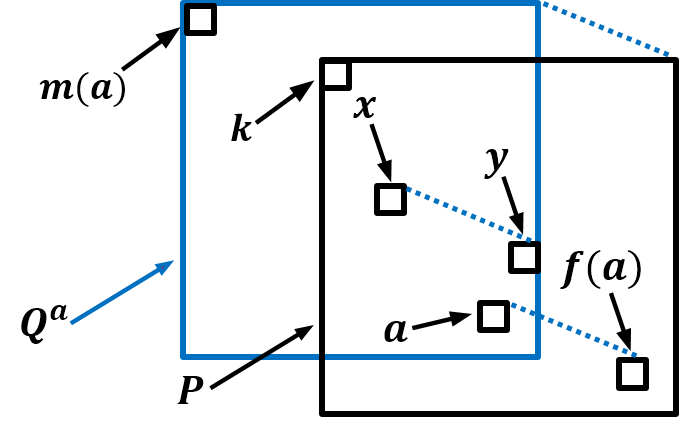}\vspace{-40pt}\\
  \end{center} \vspace{7pt}
\end{tabular}
\vspace{10pt}

\begin{definition}\label{def.mapped.to}
For three points $x, y, a \in I$, we say that $x$ is \emph{mapped} to $y$ as a result of the flipping of $a$ if $\bar{x} = m(a) + y$.
Moreover, define the function $f:I \to I$ as follows: $f(x) = \bar{x} - m(x)$ is the location to which $x$ is mapped as a result of flipping $x$.
\end{definition}

In other words, $x$ is mapped to $y$ as a result of flipping the bit $a$
if bit $\bar{x}$ of $A$ "plays the role" of bit $y$ in the new $P$-copy $Q_a$ that is created by flipping $a$.
Note that
\begin{itemize}
\item If $\bar{x} - m(a) \notin I$ then $x$ is not mapped to any point.
However, this cannot hold when $x = a$, so the function $f$ is well defined.
\item For a fixed $a$, the mapping as a result of flipping $a$ is linear: if $x$ and $y$ are mapped to $x'$ and $y'$ respectively, then $y-x = y'-x'$.
In particular, if $(x,y)$ is $(i, \Delta_i)$-related for some $1 \leq i \leq d$ then $(x', y')$ is
also $(i, \Delta_i)$-related.
\item If $x$ is mapped to $y$ as a result of flipping $a$ and $x \neq a$, then $P_x = P_y$.
\item On the other hand, we always have $P_x \neq P_{f(x)}$.
\item If $x$ is $\Delta$-internal and $(\Delta, P)$-generic, then $f(x)$ must be $(\Delta, P)$-isolated.
\end{itemize}
The first four statements are easy to verify.
To verify the last one, suppose that $x$ is internal and generic and let $z \in N_{f(x)}$; we will show that $P_{f(x)} \neq P_z$. Since $x$ is internal, there exists $y \in N_x$ such that $y - x = z - f(x)$. Then $y$ is mapped to $z$ as a result of flipping $x$, since
$\bar{y} = y + k = z + (x + k) - f(x) = z + \bar{x} - f(x) = z + m(x)$.
Therefore $P_y = P_z$. On the other hand, $P_x = P_y$ as $x$ is generic and $P_x \neq P_{f(x)}$, and we conclude that
$P_z \neq P_{f(x)}$.

\begin{lemma}\label{lem:one.isolated.point}
There is exactly one $(\Delta, P)$-isolated point in $I$.
\end{lemma}
\begin{proof}
Let $\mathcal{S}$ be the set of isolated points; our goal is to show that $|\mathcal{S}|=1$.
Consider the set
\vspace{-5pt}\[
C = \{(x,y) : x,y \in I, (x,y) \text{ is an $(i, \Delta_i)$-jump for some $1 \leq i \leq d$}\}.
\vspace{-2pt}\]
Clearly, each point in $\mathcal{S}$ is contained in at least $d$ pairs from $C$.
By claim~\ref{clm:no_isolated_neighbors} no pair of isolated points are $\Delta$-neighbours and therefore every pair in $C$ contains at most one point from $\mathcal{S}$. By Lemma \ref{lem:related_pairs}, $|C| \leq 2d$ which implies that $|\mathcal{S}|\leq 2$. On the other hand we have $|\mathcal{S}|\geq 1$.
To see this, observe that the number of $(\Delta, P)$-internal points in $I$ is greater than $\prod_{i=1}^d k/3 \geq 2^{d^2}$, while the
number of non-$\Delta$-generic points is at most $2|C| \leq 4d$, implying that at least $2^{d^2} - 4d > 0$ of the internal points are generic. Therefore, pick an internal generic point $z \in I$.
As we have seen before, $f(z)$ must be isolated.

To complete the proof it remains to rule out the possibility that $|\mathcal{S}|= 2$.
If two different $(\Delta, P)$-isolated points $a=(a_1,\ldots,a_d)$ and $b=(b_1,\ldots,b_d)$ exist, each of them must participate in exactly $d$ pairs in $C$. This implies that both of them are $\Delta$-corners with $d$ neighbors. It follows that every $\Delta$-internal point $z$ must be generic (since an internal point and a corner point cannot be neighbours), implying that either $f(z) = a$ or $f(z) = b$.

Let $1 \leq i \leq d$ and define $\delta_i > 0$ to be the smallest integer such that there exists an $(i, \delta_i)$-related pair $(x,y)$ of generic internal points with $f(x) = f(y)$.
For this choice of $x$ and $y$ we have
$m(y) - m(x) = \bar{y} - f(y) - (\bar{x} - f(x)) = \bar{y} - \bar{x} = y-x$, so $(m(x), m(y))$ is also $(i, \delta_i)$-related. In particular, we may take $\Delta_i = \delta_i$ (Recall that until now, we only used the fact that $\Delta_i < k / 3$, without committing to a specific value).
Without loss of generality we may assume that $f(x) = f(y) = a$.
By Claim \ref{clm:jumps}, any pair $(s,t)$ of $(i, \Delta_i)$-related points for which $s \neq \bar{y} - m(y) = f(y) = a$ and $t \neq \bar{x} - m(x) = f(x) = a$ is not an $(i, \Delta_i)$-jump.
Since $b$ is not a $\Delta$-neighbour of $a$, it does not participate in any $(i, \Delta_i)$-jump, contradicting the fact that it is $(\Delta, P)$-isolated. This finishes the proof of the lemma.
\end{proof}


Finally, we are ready to show that $P$ is almost homogeneous. Let $a = (a_1, \ldots, a_d)$ be the single $(\Delta, P)$-isolated point in $I$.
Consider the set
\vspace{-6pt}\[
J = \{x = (x_1, \ldots, x_d) \in I : \Delta_i \leq x_i < \Delta_i + 2^d \text{ for any $1 \leq i \leq d$}\}
\]
and note that all points in $J$ are $\Delta$-internal. Let $1 \leq i \leq d$ and partition $J$ into $(i,1)$-related pairs of points. There are $2^{d^{2}-1} \geq 4d$ pairs in the partition.
On the other hand, the number of non-generic points in $J$ is at most $2|C| - (d-1) < 4d$ (to see it, count the number of elements in pairs from $C$ and recall that $a$ is contained in at least $d$ pairs). Therefore, there exists a pair $(x,y)$ in the above partition such that $x$ and $y$ are both generic.
As before, $f(x)$ and $f(y)$ must be isolated, and thus $f(x) = f(y) = a$, implying that $\Delta_i = \delta_i = 1$. We conclude that $\Delta = (1,\ldots, 1)$.

Claim \ref{clm:jumps} now implies that any pair $(s,t)$ of $(i, 1)$-related points for which $s \neq \bar{y} - m(y) = f(y) = a$ and $t \neq \bar{x} - m(x) = f(x) = a$ is not an $(i,1)$-jump. That is, for any two neighbouring points $s,t \neq a$ in  $I$, $P_s = P_t$, implying that $P_x = P_y$ for any $x,y \neq a$ (since $\Delta = (1,\ldots, 1)$, a $\Delta$-neighbour is a neighbour in the usual sense). To see this, observe that for any two points $x,y \neq a$ there exists a path
$x_0 x_1 \ldots x_t$ in $I$ where $x_j$ and $x_{j+1}$ are neighbours for any $0 \leq j \leq t-1$, the endpoints are $x_0 = x$ and $x_t = y$, and $x_j \neq a$ for any $0 < j < t$.
Since $a$ is isolated, it is also true that $P_a \neq P_x$ for any $x \neq a$.

To finish the proof that $P$ is almost homogeneous, it remains to show that $a$ is a corner. Suppose to the contrary that $0 < a_i < k-1$ for some $1 \leq i \leq d$ and let $b,c \in I$ be the unique points such that $(a,b)$ and $(c,a)$ are $(i,1)$-related, respectively. Clearly $f(b) = a$, so $a$ is mapped to $\bar{a} - m(b) = \bar{a} - \bar{b} + f(b) = c - a + a = c$ as a result of flipping $b$, which is a contradiction - as $P_a \neq P_c$ and $b \neq a,c$. This finishes the proof.
\end{proof}


The above proof only works when the dimension is bigger than one, though it can be adapted to the one-dimensional case.
However, we present here another proof for the one-dimensional case, which is simpler than the general proof above and works for any pattern which is not almost homogeneous (as opposed to the proof above, that required the forbidden pattern to also be large enough).
The main strategy here is to consider the longest streaks of zeros and ones in the pattern - a strategy that cannot be used in higher dimensions.
\begin{theorem}[1D Modification Lemma]\label{thm:modification_1D}
A one-dimensional pattern is removable if and only if it is almost homogeneous.
\end{theorem}

\begin{proof}[Proof of Theorem \ref{thm:modification_1D}]
The reduction from a general alphabet to a binary one and the negative example for almost homogeneous patterns which were presented in the proof of Theorem \ref{thm:modification} also hold here. It remains to prove that any 1-dimensional binary pattern that is not almost homogeneous is removable.

Let $P=P_0 \ldots P_{k-1}$ be a binary pattern of length $k$, that is contained in an arbitrary binary string $S$.
We need to show that one can flip one of the bits of $P$ without creating a new $P$-copy in $S$.
We assume that $P$ contains both 0s and 1s (i.e. it is not homogeneous) otherwise flipping any bit would work. Therefore we can assume from now that $k\ge3$ (since for $k=1,2$ all patterns are homogeneous or almost homogeneous).

Let us assume also that $P$ starts with a $1$, i.e. $P_0=1$ and let $t \leq k-1$ be the length of the longest 0-streak (sub-string of consecutive 0s) in $P$. Let $i>0$ be the leftmost index in which such a 0-streak of length $t$ begins. Clearly, $P_{i-1}=1$ and $P_i= \ldots =P_{i+t-1}=0$.

If $i+t \le k$ (i.e. the streak is not at the end of $P$) then $P_{i+t} = 1$ and in such a case if we modify $P_{i+t}$ to 0, the copy of
$P$ is removed without creating new $P$-copies in $S$. To see this, observe that a new copy
cannot start at the bit flip location $i+t$ or within the 0-streak at any of its locations $i,\ldots,i+t-1$ since the bits in these locations are 0 while the starting bit of $P$ is 1. On the other hand, a new copy cannot start after $i+t$ since it must include the bit flip location or anywhere before $P_i$ since otherwise it would contain a 0-streak of length $t+1$.

This implies that $P$ contains exactly one 0-streak of length $t$ at its last $t$ locations. In particular, we have that at the last location $P_{k-1}=1$, and if we denote by $r$ the length of the longest 1-streak in $P$, a symmetric reasoning shows that $P$ begins with its only longest 1-streak of length $r$.

If $P$ is \emph{not} of the form $1^s0^t$, it can be verified that flipping $P_{s}$ (the leftmost 0 in $P$) to $1$ does not create any $P$-copy.
The only case left is $P=1^s0^t$, where $s,t \geq 2$ since $P$ is not almost homogeneous. Consider the bit of the string $S$ that is to the left of $P$. If it is a 0 then we flip $P_1$ to $0$ and otherwise, we flip $P_0$ to 0, where in both cases no new copy is created.
\end{proof}


\section{Characterizations of the Deletion Number}\label{sec.character.deletion.number}
We use the modification lemmas of Section \ref{sec.modification.lemmas} to investigate several combinatorial characterizations of the deletion number, which will in turn allow exact (and efficient) computations of the deletion number in the 1-dimensional case, as well as efficient approximation and testing of pattern freeness for removable patterns in the $d$-dimensional case for any $d$.

In particular, we prove some surprising connections between minimal deletion sets and minimal hitting sets. The characterizations for almost homogeneous  1-dimensional patterns are given in Appendix~\ref{sec:almost_homo}, along with an optimal algorithm to compute the exact deletion number and an optimal tester for pattern freeness in that case. The rest of this section deals with removable patterns, for both the 1-dimensional and multi-dimensional settings.

In the 1-dimensional case, we show that for any removable pattern there exist certain minimal hitting sets which are in fact minimal deletion sets. These are sets where none of the flips create new occurrences. Our constructive proof shows how to build such a set and allows for a linear time algorithm for finding the deletion number. The result is summarized in Theorem~\ref{thm.1D.minDeleting.equals.minHitting.Plus.Exact} and proved in Appendix~\ref{sec.character.deletion.number.proofs}.

\begin{theorem}
[$d_P(S)$ equals $h_P(S)$; Linear time computation of $d_P(S)$]
\label{thm.1D.minDeleting.equals.minHitting.Plus.Exact}
For a binary string $S$ of length $n$ and a binary pattern $P$ of length $k$ that is removable, the deletion number $d_P(S)$ equals $h_P(S)$ and can be computed in time $O(n+k)$ and space $O(k)$.
\end{theorem}


For the multidimensional case, we start by showing that when $P$ is removable,
the hitting number $h_P(A)$ of $A$ approximates the deletion number up to
a multiplicative constant that depends only on the dimension $d$. This is done in two stages, the first of which involves the analysis of a procedure that proves the existence of a large collection of $P$-copies with small pairwise overlaps, among the large set of at least $d_P(A)$ $P$-copies that exist in $A$. This procedure heavily relies on the fact that $P$ is removable.  The second stage shows the existence of a large hitting set of the collection with small pairwise overlaps.
The result is summarized in Lemma~\ref{lem:hitting_set_2D} and fully proved in Appendix~\ref{sec.character.deletion.number.proofs}.
\begin{lemma}
[relation between distance and hitting number]\label{lem:hitting_set_2D}
Let $P$ be a removable $(k,d)$-array over an alphabet $\Gamma$, and let $A$ be an $(n,d)$-array over $\Gamma$. Let $\alpha_d = 4^d + 2^d$. It holds that: $h_P(A) \leq d_P(A) \leq \alpha_d h_P(A) \leq \alpha_d (n/k)^{d}$.
\end{lemma}


\section{Testers for Pattern Freeness}\label{sec.testing.pattern.freeness}





We describe efficient testers for both the one-dimensional and the $d$-dimensional removable patterns that have tolerance and query complexity that only depend on $d$ (and not on $k$; using a completely naive tester, it can be seen that the tolerance and the query complexity depend on $k$).
The testers essentially approximate the hitting number, which is related to the deletion number by the characterizations that were shown in Section~\ref{sec.character.deletion.number}.


We start by presenting the distance approximation algorithm for $P$-freeness, which has both additive and multiplicative errors.
\begin{theorem}[Approximating the deletion number in $1$-dimension]
\label{thm:approx1D}
Let $P$ be a removable string of length $k$ and fix constants $0 < \tau < 1, 0 < \delta < 1/k$.
Let $h_1, h_2:[0,1] \to [0,1]$ be defined as $h_1(\epsilon) = (1-\tau)\epsilon - \delta$  and $h_2(\epsilon) = \epsilon + \delta$.
There exists an $(h_1, h_2)$-distance approximation algorithm for $P$-freeness
with query complexity and running time of $O(1/ k \tau \delta^2)$.
\end{theorem}
Note that $d_P(S) = h_P(S) \leq n/k$ always holds, so having an additive error parameter of $\delta \geq 1/k$ is pointless.
The proof of Theorem \ref{thm:approx1D} can be adapted to derive $(\epsilon_1, \epsilon_2)$-tolerant testers
for any $0 \leq \epsilon_1 < \epsilon_2 \leq 1$, which we describe in Theorem~\ref{thm:test_1D}. An immediate corollary is the following multiplicative tester. The proofs for Theorems~\ref{thm:approx1D} and \ref{thm:test_1D} can be found in Appendix~\ref{sec.testing.pattern.freeness.proofs}.
\begin{corollary}[Multiplicative tolerant tester for pattern freeness in $1$-dimension]
\label{cor:test_1D.multiplicative}
Fix $0 < \tau < 1$. For any $0 < \epsilon \leq 1$ there exists a $((1-\tau)\epsilon, \epsilon)$-tolerant tester
whose number of queries and running time are $O(\epsilon^{-1} \tau^{-3})$.
\end{corollary}
%
%

For the multidimensional case, our distance approximation algorithm and tolerant tester for $P$-freeness are given in Theorems \ref{thm:approx_highD} and \ref{thm:test_highD}.
As their technical details are very similar to those in the 1D case, we provide in Appendix~\ref{sec.testing.pattern.freeness.proofs} only a sketch of the main ideas.

%
%

\begin{theorem}
[Approximating the deletion number in multidimensional arrays]
\label{thm:approx_highD}
Let $P$ be a removable $(k,d)$-array and fix constants $0 < \tau \leq 1,0 \leq \delta \leq 1/k^d$.
Let $h_1, h_2:[0,1] \to [0,1]$ be defined as $h_1(\epsilon) = (1-\tau)^d \alpha_d^{-1}\epsilon - \delta$  and $h_2(\epsilon) = \epsilon + \delta$.
There exists an $(h_1, h_2)$-distance approximation algorithm for $P$-freeness
making at most $\gamma / k^d \tau^d \delta^2$ queries, where $\gamma > 0$ is an absolute constant, and has running time
$\zeta_\tau / k^d \delta^2$ where $\zeta_\tau$ is a constant depending only on $\tau$.
\end{theorem}

\begin{theorem}
[Multiplicative tolerant tester for pattern freeness in multidimensional arrays]
\label{thm:test_highD}
Fix $0 < \tau \leq 1$ and let $P$ be a removable $(k,d)$-array. For any $0 < \epsilon \leq 1$ there exists a $((1-\tau)^d \alpha_{d}^{-1}\epsilon, \epsilon)$-tolerant tester making $C_\tau \epsilon^{-1}$ queries, where $C_\tau = O(1 / \tau^d (1 - (1-\tau)^d)^2)$.
The running time is $C'_\tau \epsilon^{-1}$ where $C'_\tau$ depends only on $\tau$.
\end{theorem}


%
%

\section{Discussion and Open Questions}
We have provided efficient algorithms for testing whether high-dimensional arrays do not contain a fixed pattern $P$ for any removable pattern $P$.
The results suggest several interesting open questions on the problem of pattern-freeness and more generally, on local properties - where we say that a property $\mathcal{P}$ is $k$-local if any array $A$ not satisfying $\mathcal{P}$, there exists a consecutive subarray of $A$ of size at most $k \times \ldots \times k$ which does not satisfy $\mathcal{P}$ as well. That is, a property is local if any array not satisfying $\mathcal{P}$ contains a small `proof' for this fact. Note that $P$-freeness is indeed $k$-local where $k$ is the side length of $P$, and that a property $\mathcal{P}$ is $k$-local if and only if there exists a \emph{family} $\mathcal{F}$ of arrays of size at most $k \times \ldots \times k$ each, such that $A$ satisfies $\mathcal{P}$ if and only if it does not contain any consecutive sub-array from $\mathcal{F}$. That is, to understand the general problem of testing local properties of arrays we will need to understand the testing of $\mathcal{F}$-freeness, where $\mathcal{F}$ is a family of forbidden patterns (rather than a single forbidden pattern).

In particular, the problem of approximate pattern matching is of interest. The family of forbidden patterns for this problem might consist of a pattern and all patterns that are close enough to it, and the distance measures between patterns might also differ from the Hamming distance (e.g., $\ell_1$ distance for grey-scale patterns).

Finally, it is desirable to settle the problem of testing pattern freeness for the almost homogenous case by either finding an efficient tester for the almost homogeneous multi dimensional case, or proving that an efficient tester cannot exist for such patterns. It is also of interest to examine which of the $[k]^d$ patterns with $k<3 \cdot 2^d$ are removable.

%
%


\vspace{-8pt}
\subparagraph*{Acknowledgements}
We are grateful to Swastik Kopparty for numerous useful comments. We are thankful to Sofya Raskhodnikova for her useful feedback.
\vspace{-6pt}


\appendix

\section{Characterizations of the Deletion Number: Proofs}\label{sec.character.deletion.number.proofs}

\begin{proof}[Proof of Theorem \ref{thm.1D.minDeleting.equals.minHitting.Plus.Exact}]

The main challenge is in proving that $d_P(S)=h_P(S)$, since then all we need is an algorithm that computes $h_P(S)$, which is relatively standard in template matching: Find the set ${\cal O}$ of all $P$-copies in $S$; Go though the $P$-copies in ${\cal O}$ from left to right, repeating the following: (i) Let $P^*$ be the leftmost $P$-copy in ${\cal O}$; (ii) Increment the hitting set count by 1; (iii) Remove from ${\cal O}$ all the (following) $P$-copies that intersect $P^*$ (those whose starting location is not to the right of the rightmost location in $P^*$);. Clearly, the complexity of the algorithm is dominated by the first step of finding ${\cal O}$, which can be done in $O(n+k)$ using, e.g., the KMP algorithm~\cite{knuth}. Taking the rightmost location in each of the visited $P^*$s creates a hitting set, which is minimal, due to the fact that the set of $P^*$s is independent.

It is trivial that $d_P(S)\ge h_P(S)$ and hence we have to show that $d_P(S)\le h_P(S)$. Refer to Algorithm~\ref{alg.efficient.removal.NAH} below that constructs a set of bit flip locations. Note that the choice in Step 3 is possible using the modification lemma, while the choice in Step 4 is possible, since if $h$ is contained in only one $P$-copy $P^0\in{\cal D}$, by definition of ${\cal D}$ there is some $P^1\in{\cal D}$ such that $P^0$ and $P^1$ intersect at some location $x$ (in particular one of the 2 endpoints of $P^0$ must be in the intersection). Simply replace $h$ by $x$. It is easy to verify that the set of locations ${\cal F}$ that it computes is a (particular) minimal hitting set of ${\cal O}$, and hence ${|\cal F}|=h_P(S)$. It is therefore sufficient to show that flipping the bit locations in ${\cal F}$ turns the string $S$ to be $P$-free. This will be guaranteed, using the fact that ${\cal F}$ is a hitting set of ${\cal O}$, by Lemma~\ref{lem.removable.flipping.F} that shows that no bit flip of a location in ${\cal F}$ creates a new $P$-copy. Therefore, he proof of Lemma~\ref{lem.removable.flipping.F} will complete the proof of Theorem~\ref{thm.1D.minDeleting.equals.minHitting.Plus.Exact}.

\vspace{-0.5pt}
\begin{algorithm}\caption{}
\label{alg.efficient.removal.NAH}\vspace{2pt}
\textbf{Input:} \raggedright{Binary string $S$ of length $n$ and removable binary string $P$ of length $k$}
\\ \vspace{2pt}
\textbf{Output:} \raggedright{Minimal set ${\cal F}$ of flip locations in $S$ that make it $P$-free ($|F|=d_P(S)$)}
\\ \vspace{2pt}
\vspace{-3pt}
  \begin{enumerate}
  \vspace{-5pt} \item Find the set ${\cal O}$  of all $P$-copies in $S$
  \vspace{-1pt} \item Divide ${\cal O}$ into ${\cal I} \cup {\cal D}$, where ${\cal I}$ is the subset of $P$-copies that do not intersect any other $P$-copy in ${\cal O}$, while ${\cal D}$ is the subset of $P$-copies that intersect some other $P$-copy in ${\cal O}$.
  \vspace{-1pt} \item For each $P$-copy $P^*\in{\cal I}$ add to ${\cal F}$ a bit location whose flipping removes $P^*$ without creating any other $P$-copy
  \vspace{-1pt} \item Find a minimal hitting set ${\cal H}$ of ${\cal D}$ such that every location $h\in{\cal H}$ is contained in at least two $P$-copies in ${\cal D}$.
  \vspace{-1pt} \item Add ${\cal H}$ to ${\cal F}$
\end{enumerate}  \vspace{-7pt}
\textbf{return} ${\cal F}$\vspace{2pt}
\end{algorithm}

\begin{lemma}
[Flipping bits in $F$ does not create new $P$-copies]
\label{lem.removable.flipping.F}
Let $f\in{\cal F}$. Flipping the bit at location $f$ does not create any new $P$-copy in $S$.
\end{lemma}

\begin{proof}
Recall that ${\cal F}$ consisted of bits in ${\cal I}$ as well as bits in ${\cal D}$. Each of the bit flips that are in ${\cal I}$ was chosen (step 3 of Algorithm~\ref{alg.efficient.removal.NAH}) using the modification lemma to be such that no new $P$-copy is created.

The main challenge is in showing that the remaining bit flips, i.e. at locations ${\cal H}$, do not create any new $P$-copies. Notice our requirement that any location $h\in{\cal H}$ is contained in at least two $P$-copies. By symmetry considerations, we have the following

\begin{observation}\label{obs.symmetric.flipping}
  $\textnormal{[}$Flipping an arbitrary bit in the intersection of 2 $P$-copies can create a new $P$-copy$\textnormal{]}$
  $\Longleftrightarrow$
  $\textnormal{[}$Flipping an arbitrary bit in a $P$-copy can create 2 new $P$-copies$\textnormal{]}$
\end{observation}

By Observation~\ref{obs.symmetric.flipping}, in order to show that bit flips in ${\cal H}$ do not create new $P$-copies, one can prove that an arbitrary bit-flip in a $P$-copy cannot create more than 1 $P$-copy, as is stated in the next lemma.

\end{proof}

\begin{lemma}
[Any bit flip in a pattern $P$ cannot create more than 1 new $P$-copy]\label{lem.removable.bit.flips}
Let $x\in [k]$. Flipping the bit $P_x$ can create at most 1 new $P$-copy in $S$.
\end{lemma}

\begin{proof}
The proof goes by contradiction, assuming that a bit flip in $P$ has created two new $P$-copies $P^1$ and $P^2$, and will analyze separately the two possible cases:

\ourParagraph{case 1: `$P^1$ and $P^2$ intersect $P$ from different sides'}
In this case, flipping the bit location $x$ of $P$ creates a $P$-copy $P^1$ shifted $t_1$ locations to the left and a $P$-copy $P^2$ shifted $t_2$ locations to the right, where we assume w.l.o.g. that $t_1<t_2$. One can verify that $P_{x-t_2}=P^2_{x-t_2}\neq P_x$ (and similarly that $P_x\neq P^1_{x+t_1}=P_{x+t_1}$). We will assume that $P_x=0$ and hence $P_{x-t_2}=1$. We refer the reader to Figure~\ref{fig.2copies_diff_side} and its caption for the intuition of the proof.

\begin{figure} [h!] \vspace{18pt}
\begin{center}
\small\addtolength{\tabcolsep}{0pt}
\hspace{-12pt}
\begin{tabular}{p{0.47\textwidth}p{0.41\textwidth}}
\vspace{-122pt}
\caption{\textbf{Illustration for case 1:} Our proof is based on 'skipping' along a 'path' from location $x$ to location $x-t_2$ in $P$, while each skip is done between entries with equal values. A complete path from $x$ to $x-t_2$ will give a contradiction, since $P_{x-t_2}\neq P_x$. The path starts at $x$ and makes skips of size $t_1$ to the left as long as it does not pass $x-t_2$, then it makes a single skip to the right of size $t_2$. It repeats this traversal until reaching $x-t_2$.}\label{fig.2copies_diff_side}
&
\includegraphics[width = 0.45\textwidth]{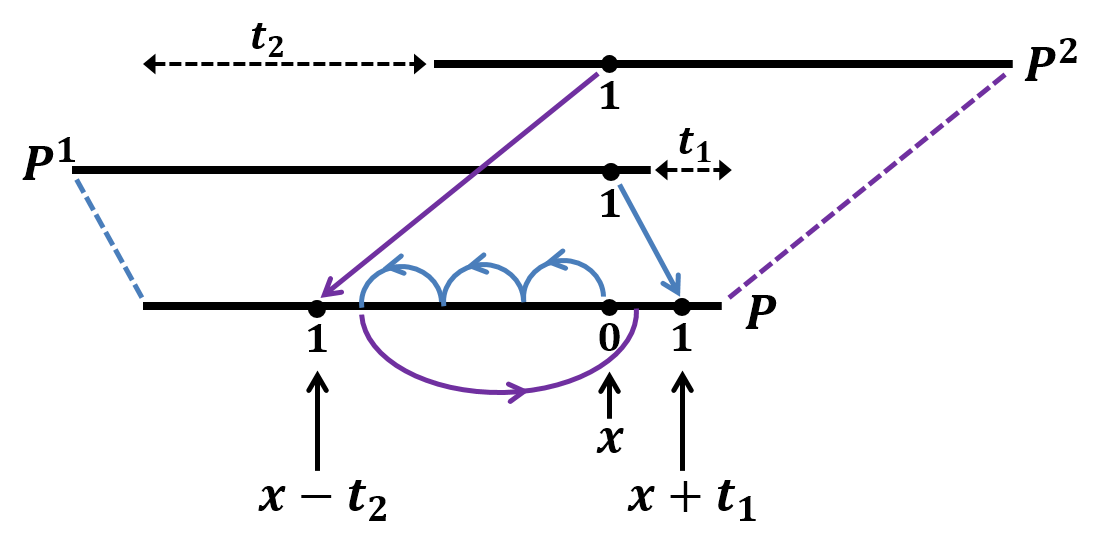}
\end{tabular}
\end{center}
\vspace{-24pt}
\end{figure}

Since the $P$-copy $P^1$ was created from $P$ at a left offset of $t_1$ by the flipping at location $x$, we can infer that $P_{y}=P_{y+t_1}$ for any $y\in[k-t_1]$ , $y\neq x$ (or informally that "$P$ is $t_1$-cyclic except at $x$ from the right" ). Similarly, we know that "$P$ is $t_2$-cyclic except at $x$ from the left".

We define a 'path' of skips that starts from location $x$, makes skips of size $t_1$ to the left as long as it does not pass $x-t_2$, then it makes a single skip to the right of size $t_2$. Call this short path a traversal. The path repeats this traversal until reaching $x-t_2$. It is easy to verify that the path is always within the open range $(x-t_2,x+t_1)$ (except for the last step that reaches $x-t_2$). This implies in particular that the path does not go from $x+t_1$ to $x$ or from $x$ to $x-t_2$ (i.e. through the two only "value switching skips"), and hence the value of $P$ along the path must be 0.

It remains to prove that the path eventually reaches $x-t_2$ and does not continue in some infinite loop. For each location $y$ that the path goes through we can look at the value $y \Mod{t_1}$. Assume w.l.o.g. that for the 'target' location $x-t_2$ we have that $x-t_2 = 0 \Mod{t_1}$. This implies for the 'starting' location $x$ that $x = t_2 \Mod{t_1} = 0$. Now, each skip by $t_1$ does not change the location $\Mod{t_1}$, while a skip by $t_2$ to the right increases the value by $t_2 \Mod{t_1}$. In other words, the sub-sequence of locations at the beginning of each traversal (before the first left skip) is of the form $\ell\cdot t_2 \Mod{t_1}$, for $\ell = 1,2,3,\ldots$ . This is exactly the subgroup of $\mathbb{Z}_{t_1}$ (the additive group of integers modulo $t_1$) generated by the element $t_2$ and hence must contain the identity element $0 \Mod{t_1}$. This proves that the location $x-t_2$ will be reached.

\ourParagraph{case 2: `$P^1$ and $P^2$ intersect $P$ from one (the same) side'}
Flipping a location $x$ in a $P$-copy $P$ creates two new $P$-copies $P^i$ ($i=1,2$) that intersect $P$ from the same side, w.l.o.g. right, at a shift of $t_i$, where $t_1<t_2$. Refer to Figure~\ref{fig.2copies_same_side} and its caption for the intuition of the proof. We 'follow' the two disjoint 'arrow paths' shown in the figure that lead from $x$ in $P$ to $x':=x-t_2$ in $P^1$ to reach a contradiction. Formally:
$$P_x=P^1_x=P^2_{x-t_2+t_1}=P_{x-t_2+t_1}=P^1_{x-t_2}$$
$$P_x\neq P^2_{x-t_2}=P_{x-t_2}=P^1_{x-t_2}$$

\begin{figure} [h!] \vspace{4pt}
\begin{center}
\small\addtolength{\tabcolsep}{8pt}
\hspace{-18pt}
\begin{tabular}{p{0.41\textwidth}p{0.4\textwidth}}
\vspace{-93pt}
\caption{\textbf{Illustration for case 2:} All arrows (ignoring directions) except the red one represent equality, while the red arrow represents inequality. The two disjoint 'arrow paths' from $x$ in $P$ to $x'$ in $P^1$ imply that both $P_x=P_{x'}$ and $P_x\neq P_{x'}$, leading to contradiction.}\label{fig.2copies_same_side}
&
\includegraphics[width = 0.45\textwidth]{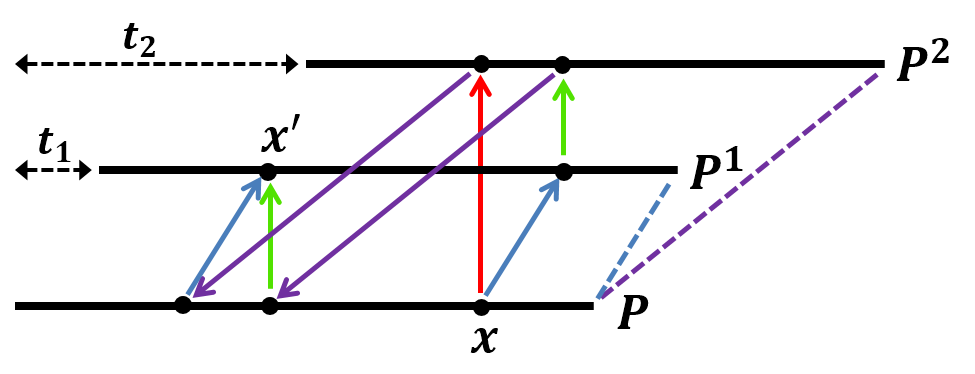}
\end{tabular}
\end{center}
\vspace{-24pt}
\end{figure}
\end{proof}
\end{proof}

\begin{proof}[Proof of Lemma~\ref{lem:hitting_set_2D}]
The first inequality follows from the fact that one needs to modify at least one entry in any $P$-copy in $A$.
For the third inequality, note that the set $\{(x_1, \ldots, x_d) \in [k]^d : \forall 1 \leq i \leq d,\  x_i \equiv k-1 (\text{mod }{k})\}$ is a set of size $[n/k]^d$ that hits all $k \times \ldots \times k$ consecutive subarrays of $A$, and in particular all $P$-copies.
It remains to prove that $d_P(A) \leq \alpha_d h_P(A)$.
We may assume that the alphabet $\Gamma$ is binary by applying the standard reduction from non-binary to binary alphabets presented in Section \ref{sec.modification.lemmas}.

We present a procedure on the array $A$ that makes it $P$-free by sequentially flipping bits in it.
In what follows, we will say that the \emph{center} of a $(k,d)$ matrix lies in location $(\lfloor k/2 \rfloor, \ldots, \lfloor k/2 \rfloor)$ in the matrix.
Let $\cal \mathcal{P}$ be the set of all $P$-copies \emph{before} $A$ is modified.
In Phase 1, the procedure "destroys" all $P$-copies in $\mathcal{P}$ by flipping central bits of a subset of the original $P$-copies in $A$, which is chosen in a greedy manner. However, these bit flips might create new $P$-copies in $M$, which are removed in Phase 2 using the modification lemma.
The procedure maintains sets ${\cal A}, {\cal B}$ that contain the bits flipped in phases 1,2 respectively.

\begin{itemize}\vspace{-5pt}
  \item Let ${\cal P}$ be the set of all $P$-copies in $A$, ${\cal N} \leftarrow \phi$ ${\cal A} \leftarrow\phi$, ${\cal B}\leftarrow\phi$.
  \item \textbf{Phase 1: }While ${\cal P} \neq \phi$
  \begin{itemize}
    \item Pick $Q \in {\cal P}$ arbitrarily.
    \item Flip $A_x$ where $x$ is the center of $Q$.
    \item Add $Q$ to ${\cal A}$ and remove all $P$-copies containing $x$ from ${\cal P}$.
    \item Add all $P$-copies \emph{created} by flipping $A_x$ to $\cal N$.
    \end{itemize}
  \item \textbf{Phase 2: }While ${\cal N} \neq \phi$
  \begin{itemize}
    \item Pick $Q \in {\cal N}$ arbitrarily.
    \item Pick a location $x$ in $Q$ whose flipping does not create new $P$-copies in $A$ (exists by modification lemma).
    \item Flip the bit $A_x$ and add $x$ to ${\cal B}$.
  \end{itemize}
\end{itemize}
For the analysis of the procedure, we say that two $P$-copies $Q,Q'$ in $A$ whose starting points are $x = (x_1, \ldots, x_d), y = (y_1, \ldots, y_d) \in [n]^d$ respectively are \emph{$1/2$-independent} if $|x_i - y_i| \geq k/2$ for some $1 \leq i \leq d$.
Note that $1/2$-independence is a symmetric relation.
A set of $P$-copies is $1/2$-independent if all pairs of copies in it are $1/2$-independent. Denote by $i_{P}(A)$ the maximal size of a $1/2$-indpendent set in $A$, divided by $n^d$.

For $Q$ and $Q'$ as above, if $Q'$ does not contain the center of $Q$ then $Q,Q'$ are $1/2$-independent, as
there is some $1 \leq i \leq d$ for which either $y_i < x_i + \lfloor k/2 \rfloor - (k-1) \leq x_i - k/2 + 1$ or $y_i > x_i + \lfloor k / 2 \rfloor \geq x_i + k/2 -1$.
In both cases $|y_i - x_i| \geq k/2$, implying the $1/2$-independence.
Therefore the set ${\cal A}$ generated by the procedure is $1/2$-independent: if $Q, Q' \in {\cal A}$ are two different $P$-copies and $Q$ was added to $\cal A$ before $Q'$, then $Q'$ does not contain the center of $Q'$, so $Q$ and $Q'$ are $1/2$-independent.
Using the following claims, it is not hard finish the proof of the lemma.
\begin{claim}
\label{clm:independece_hitting}
$i_P(A) \leq 2^d h_P(A)$.
\end{claim}
The proof of Claim \ref{clm:independece_hitting} will be given later.
For what follows, we say that
$P$ has a \emph{cycle} of size $t=(t_1,\ldots,t_d)\in\mathbb{Z}^d$, if $P_x=P_y$ for every pair of locations $x=(x_1,\ldots,x_d),y=(y_1,\ldots,j_d)\in[k]^d$ such that $x_i\equiv y_i \Mod{|t_i|}$ $\forall i\in[d]$.
The following claim is straightforward to verify.
\begin{claim}\label{obs:overlap_to_cyclic_2D}
\textnormal{{[Shifted occurrences imply a cyclic pattern]}}
If $M$ contains two overlapping occurrences of $A$, at a relative offset of $t\in\mathbb{Z}^d$, then $P$ has a cycle of size $t$.
\end{claim}

\begin{claim}\label{clm:Bit_flip_no_new_occurrences_2D}
\textnormal{{[Central bit flip creates few new occurrences]}}
Flipping the central bit of a $P$-occurrence in $A$ creates at most $2^d$ new occurrences of $P$ in $A$.
\end{claim}

We first show how to use these claims to finish the proof.
Consider the sets $\mathcal{A}, \mathcal{B}$ after the procedure ends.
The procedure flips $|{\cal A}| + |{\cal B}|$ bits in $A$, so $|{\cal A}| + |{\cal B}| \geq d_P(A) n^d$.
On the other hand, $|\mathcal{A}| \leq i_P(A) n^d \leq (2n)^d h_P(A)$ as $\cal A$ is $1/2$-independent.
Claim \ref{clm:Bit_flip_no_new_occurrences_2D} now implies that $|{\cal B}| \leq 2^d|{\cal A}|$, and we get that
\[
n^d d_P(A) \leq |{\cal A}| + |{\cal B}| \leq (2^d+1)|\mathcal{A}| \leq \alpha_d  n^d h_P(A)
\]
Dividing by $n^d$ yields the desired inequality. We now prove the claims.

\begin{proof}[Proof of Claim \ref{clm:independece_hitting}]
Let $\cal S$ be a $1/2$-independent set of $P$-copies in $A$, which is of size $i_P(A)$.
We will show that no point in $[n]^d$ is contained in more than $2^d$ copies from $S$, implying that to hit all copies of $P$ in $A$
(and in particular, all copies of $P$ in $\cal S$) we will need at least $|{\cal S}| / 2^d = i_P(A) / 2^d$ entries.

Suppose to the contrary that there are $2^d + 1$ copies from $\cal S$ that contain the point $x = (x_1, \ldots, x_d) \in [n]^d$.
we will say that a copy from $\cal S$ containing $x$ is $i$-lower if $k/2 \leq x_i - y_i < k$
and $i$-higher if $0 \leq x_i - y_i < k/2$ (note that $0 \leq x_i - y_i < k$ must hold).
Therefore, there exist two copies $Q, Q' \in {\cal S}$ containing $x$, starting at $(y_1, \ldots, y_d)$ and $(y'_1, \ldots, y'_d)$ respectively, such that for any $1 \leq i \leq d$, $Q$ is $i$-higher ($i$-lower)
\emph{if and only if} $Q'$ is $i$-higher ($i$-lower respectively).
But then, for any $i$, either $0 \leq y_i, y'_i < k/2$ or $k/2 \leq y_i, y'_i < k$, implying that $|y_i - y'_i| < k/2$, thus contradicting the fact that $\cal S$ is $1/2$-independent.
\end{proof}

\begin{proof}[Proof of Claim \ref{clm:Bit_flip_no_new_occurrences_2D}]
Assume that more than $2^d$ new occurrences are created. Since these occurrences overlap (at the bit flip location),the same argument as in Claim \ref{clm:independece_hitting} implies that there must be two of them, $P_1$ and $P_2$, that are shifted (one from the other) by some vector $t\in\mathbb{Z}^d$, where $|t_i|<k/2$ $\forall i\in[d]$.

By Claim \ref{obs:overlap_to_cyclic_2D}, $P$ (and hence also $P_1$ and $P_2$) has a cycle of size $t$.
Let $x$ be the point in $M$ of the (central) flipped bit in $P_0$ and consider the point $x'=x+t$, which is also in $P_0$, since $|t_i|< k/2$ $\forall 1\le i \le k$. The occurrence $P_2$ overlaps both locations $x$ and $x'$ (since both new occurrences $P_1$ and $P_2$ overlap the bit flip location $x$ and $P_2$ is shifted by $t$ from $P_1$, which overlaps $x$).

On one hand we have $M_x=M_{x'}$ (before the bit flip), since both locations belong to $P_0$, which has a cycle of size $t$. On the other hand, $M_x\neq M_{x'}$, since these locations both belong to $P_2$ and must be equal \emph{after} flipping $M_x$ as $P_2$ has a cycle length of $t$. This leads to a contradiction.
\end{proof}

\end{proof}

\section{Testers for Pattern Freeness: Proofs}\label{sec.testing.pattern.freeness.proofs}

\begin{theorem}
\label{thm:test_1D}
Let $P$ be a removable string of length $k$ and let $0 \leq \epsilon_1 < \epsilon_2 \leq 1$. There exists an $(\epsilon_1, \epsilon_2)$-tolerant tester
whose number of queries and running time are $O(\epsilon_2^2 / (\epsilon_2 - \epsilon_1)^3)$ where the constant term does not depend on $k$.
\end{theorem}
It is not clear whether this upper bound is tight in general.
However, for the important special case of tolerant testers with multiplicative tolerance of $1 + \tau$,
where $\tau > 0$ is a constant, the above tester is optimal (up to a multiplicative constant that depends on $\tau$), as is
shown by taking $\epsilon_2 = \epsilon$ and $\epsilon_1 = (1-\tau)\epsilon$ in Theorem \ref{thm:test_1D}, leading to the multiplicative tester given in Corollary~\ref{cor:test_1D.multiplicative}.

\begin{proof}[Proof of Theorems \ref{thm:approx1D} and \ref{thm:test_1D}]
Let $S$ be a string of length $n \geq \beta k$, where $\beta = 3 / \tau$. Write $\epsilon = \delta_P(S)$
and let $H \subseteq [n]$ be a hitting set for $P$ in $S$ whose size is $\epsilon n$. That is, $H$ is a minimal set of locations that satisfies the following: if $S$ contains a copy of $P$ starting at location $l$, then
$\{l, \ldots, l+k-1\} \cap H \neq \phi$.

For $i \in [n]$ let $I_i$ denote the ``cyclic interval'' of length $\beta k$ starting at $i$. That is, if $i + \beta k > n$ then $I_i = \{i, \ldots, n\} \cup \{0, \ldots, i + \beta k - n - 1\}$ and otherwise $I_i = \{i, \ldots, i + \beta_k - 1\}$. 

Let the random variable $X$ denote the
size of the minimal hitting set $H_i$ for $P$ in the interval $I_i$, divided by $\beta k$, where $i \in [n]$ is chosen uniformly at random.
Note that $X$ is computable in time $O(\beta k)$, by Theorem \ref{thm.1D.minDeleting.equals.minHitting.Plus.Exact}.
Let $\mu$ and $\sigma^2$ denote the expectation and the variance of $X$, respectively.
By the minimality of $H_i$, we have that $|H_i| \leq |H \cap I_i|$ since the set in the RHS is a hitting set for $P$ with respect to the interval $I_i$.
Thus, $\mu \leq \mathbb{E}[|H \cap I_i|] / \beta k = \epsilon$.

Next we bound $\mu$ from below.
Since $H_i$ hits all $P$-copies that lie exclusively inside $I_i$, and by the minimality of $H$, we must have
$|H_i| > |H \cap I'_i|$ where $I'_i$ is the cyclic interval that starts in $i+k$ and ends in $(i + (\beta-1)k - 1) \mod{n}$.
Therefore, $\mu \geq \mathbb{E}[H \cap I'_i] / \beta k = (1-2/\beta) \epsilon \geq (1-\tau)\epsilon$.
To conclude, we have seen that $(1 - \tau)\epsilon \leq \mu \leq \epsilon$.


To compute the variance of $X$, note that $0 \leq X \leq 1/k$, as there exist $\beta$ entries in $I_i$ such that any subinterval of length $k$ in $I_i$ contains at least one of them. By convexity, the variance satisfies $\sigma^2 \leq k\mu (1/k - \mu)^2 + (1 - k\mu) (0 - \mu)^2 = \mu (1/k-\mu) \leq \epsilon / k$.

Now let $Y = \frac{1}{t}\sum_{j=1}^{t} X_j$ where the $X_j$ are independent copies of $X$ and $t$ will be determined later. Then $\mathbb{E}[Y] = \mu$ and $\text{Var}(Y) = \sigma^2 / t  \leq \epsilon / k t$.

Recall that $\mathbb{E}[Y] = \mu$, where $(1-\tau) \epsilon \leq \mu \leq \epsilon$, so to get the desired approximation, it suffices to estimate $Y$
with an additive error of no more than $\delta$ with constant probability. Chebyshev inequality implies that it suffices to have $Var(Y) = \Theta(\delta^2)$. In other words, it will be enough to sample
$t = \Theta(\epsilon /k \delta^2)$ blocks, each of size $\beta k = \Theta(k / \tau)$. In total, it is enough to make $\Theta(k\epsilon (1/k-\epsilon) / \tau \delta^2) = O(\epsilon/\tau \delta^2)$ queries.

In the setting of approximation, $\epsilon$ is not known in advance, but $\epsilon \leq 1/k$ always holds,
so sampling $t = \Theta(1/k^2 \delta^2)$ blocks would suffice to get the desired additive error.
The return value of the approximation algorithm will be its estimate of $Y$.
The query complexity and running time are $\beta t k = \Theta(1 / k \tau \delta^2)$.
This finishes the proof of Theorem \ref{thm:approx1D}.

Now consider the setting of $(\epsilon_1, \epsilon_2)$-tolerant testing.
By monotonicity of the tester, we can assume that we are given a string whose distance from $P$-freeness is either exactly $\epsilon_2$ or
exactly $\epsilon_1$.
Pick $\epsilon = \epsilon_2$, $\delta = (\epsilon_2 - \epsilon_1)/4$, $\tau = (\epsilon_2 - \epsilon_1) / 4 \epsilon_2$,
and sample $t = \Theta(\epsilon /k \delta^2)$ blocks, with query complexity and running time
$\Theta(\epsilon /\tau \delta^2) = \Theta(\epsilon_2^2 / {\epsilon_2 - \epsilon_1}^3)$, as was stated above.
If the given string $S$ is $\epsilon_2$-far from $P$-freeness, then with probability at least $2/3$,
after sampling $t = \Theta(\epsilon /\tau \delta^2)$ samples, the value of $Y$ will be bigger than $(\epsilon_2)*(1 - \tau) - \delta = (\epsilon_2 + \epsilon_1)/2$.
On the other hand, if $S$ is $\epsilon$-close then with probability at least $2/3$, $Y \leq \epsilon_1 + \delta < (\epsilon_2 + \epsilon_1)/2$
Therefore, the tester will answer that the input is $\epsilon_2$-far if and only if $Y \geq (\epsilon_2 + \epsilon_1)/2$.
This finishes the proof of Theorem \ref{thm:test_1D}.

%
%
\end{proof}

\begin{proof}[Proof sketch for Theorems \ref{thm:approx_highD} and \ref{thm:test_highD}]
Take $\beta = 2 / \tau$.
Let $A$ be an $(n,d)$-array where we may assume that $n \geq \beta k$ for a suitable choice of $C$.
Again, the strategy is to take $t$ (to be determined) independent samples of blocks of size $\beta k \times \ldots \times \beta k$
and compute the hitting number of each sampled block.
Note that (as opposed to the one-dimensional case), computing the minimal hitting set is generally an $NP$-complete problem, but
since the hitting number of each of these blocks is at most $\beta^d = \Theta(\tau^{-d})$, here we may compute it
with running time that depends only on $\tau$ and $d$.
As in the 1D case, the expected relative hitting number $\mu$ of a sampled block satisfies $(1-\tau)^d h_P(A) = (1-2/\beta)^d h_P(A) \leq \mu \leq h_P(A)$.
The variance of the hitting number for a single sample is no bigger than $k^d (1/k^d - \mu)^2 + (1 - k^d \mu)\mu^2 = \mu(1/k^d-\mu) \leq \mu / k^d$, so for $t$ samples it is $O(h_P(A) / k^d t)$.
To get additive error of at most $\delta$ with constant probability, we may have (by Chebyshev inequality) $h_P(A) / k^d t = \Theta(\delta^2)$,
or $t = \Theta(h_P(A) / k^d \delta^2)$.

Therefore, for an approximation algorithm (in which we don't know $h_P(A)$ in advance, though we have an upper bound of $h_P(A) \leq 1 / k^d$), $t = \Theta(k^{-2d} \delta^{-2})$ sampled blocks are enough,
and the total number of samples is $O(1 / k^d \tau^d \delta^2)$.
For a $((1-\tau)^d\epsilon, \epsilon)$-tolerant tester for the hitting number (which translates to a $((1-\tau)^d \alpha_d^{-1} \epsilon, \epsilon)$-tolerant tester for the deletion number), as observed in the 1D case, when deciding on the number of samples we may assume that $h_P(A) = \epsilon$ and pick $\delta  = \Theta((1-(1-\tau)^d) \epsilon)$ , so $t = \Theta(\epsilon / k^d \delta^2) = \Theta(1 / k^d (1 - (1-\tau)^d)^2 \epsilon)$ sampled block suffice. Since each block is of size $\Theta(k^d / \tau^d)$, the total number of queries is $O(C_\tau \epsilon^{-1})$ where $C_\tau = 1 / \tau^d (1 - (1 - \tau)^d)^2$, while the running time is $C'_\tau \epsilon^{-1}$, where $C'_\tau$ depends on the time required to compute the hitting number in a single sampled block.
\end{proof} 


\section{Almost Homogeneous Patterns}
\label{sec:almost_homo}
The testers discussed above only consider removable patterns.
However, Theorem \ref{thm:modification} states that almost homogeneous patterns
over a binary alphabet are not removable. As the above testers are not applicable
for this case, it is natural to ask whether there exist efficient testers (and in particular tolerant testers)
for pattern freeness when the pattern is almost homogeneous.

In this subsection we partially answer this question, addressing the one-dimensional case; we do so by
describing a simple yet powerful characterization of the distance to $P$-freeness when $P$ is almost homogeneous.
The characterization is then utilized to
get both an optimal $1/(16+c)$-tolerant tester for $P$-freeness in the one-dimensional case and an efficient \emph{exact}
algorithm for computing the distance to $P$-freeness.

Unfortunately, this characterization does not hold in higher dimensions, and the question of
testing and approximating the distance to $P$-freeness in this case is left open.
It is not clear whether there exist efficient algorithms for computing the exact distance to $P$-freeness in high dimensions.
It will be interesting to either find such an algorithm or show that the problem is NP-hard.

%
Denote the absolute distance of a string $S$ to $P$-freeness by $d_P(S)$.
Our main results here are as follows.
\begin{theorem}
\label{thm:compute_almost_homo}
There exists an algorithm that, given an almost homogeneous binary string $P$ and a binary string $S$ of length $n$, compute $d_P(S)$ in time $O(n)$ and space $O(1)$ where the constants do not depend on the length of $P$.
\end{theorem}
The above algorithm is clearly optimal (up to a multiplicative constant that does not depend on $k$) with respect to time and space complexity.
Our second result is a $1/(16+c)$-tolerant tester for $P$-freeness when $P$ is almost homogeneous.
\begin{theorem}
\label{thm:test_almost_homo}
Let $P$ be an almost homogeneous binary string and fix a constant $c > 0$. Then there exists a $1/(16+c)$-tolerant $\epsilon$-tester for $P$-freeness
making at most $\alpha_c \epsilon^{-1}$ queries, where $\alpha_c$ depends only on $c$.
\end{theorem}
Again, the running time of the above tester is optimal (up to a multiplicative constant).
In particular, it does not depend on the length of the forbidden pattern $P$ or the tested string.
We did not try to optimize the tolerance factor in the above statement.

In the rest of this subsection we will prove the above two theorems. We will start by describing a characterization
of the distance to $P$-freeness which is easier to work with.
Without loss of generality, in the rest of the subsection assume that $P = 10^{k-1}$ for some $k > 0$ and that $S$ is a string of length $n$.
A pair $(i,j)$ with $0 \leq i < j \leq n-k+1$ is said to be a \emph{$P$-evidence} in $S$ if $P_i = 1$ and $P_l = 0$ for any $j \leq l \leq j+k-2$.
\begin{observation}
\label{obs:distance_evidence}
$S$ is $P$-free \emph{if and only if} there are no $P$-evidences in $S$.
\end{observation}
\begin{proof}
If $S$ contains a $P$-copy at location $i$ then the pair $(i,i+1)$ is a $P$-evidence.
On the other hand, suppose that $(i,j)$ is a $P$-evidence and let $i' < j$ be the maximal integer such that
$P_{i'} = 1$. Then $P_{l} = 0$ for any $i' < l \leq j+k-2$, implying that $S$ contains a $P$-copy starting at $i'$.
\end{proof}
The main technical result of this subsection relates the distance to $P$-freeness to the maximal number of non-overlapping $P$-evidences, where
two $P$-evidences $(i,j)$ and $(i',j')$ are non-overlapping if $i \neq i'$ and $|j - j'| \geq k-1$.
Denote the maximal number of non-overlapping $P$-evidences in $S$ by $d_E(S)$.
\begin{lemma}
\label{lem:many_disjoint_almost_homo}
$d_P(S) = d_E(S)$.
\end{lemma}
\begin{proof}[Proof of Lemma \ref{lem:many_disjoint_almost_homo}]
It is clear that $d_P(S) \geq d_E(S)$. Indeed, take a set of non-overlapping $P$-evidences which is of size $d_E(S)$, then by Observation \ref{obs:distance_evidence}
we will need to modify at least one bit corresponding to any evidence, and in total at least $d_E(S)$ bits in $S$, to make it $P$-free.
Next we show that $d_P(S) \leq d_E(S)$ by showing that there exists a set $T$ of non-overlapping $P$-evidences in $S$ with the following property:
to make $S$ $P$-free, it is enough to modify exactly one bit in any evidence in $T$.

The construction of $T$ is carried by ``marking'' entries corresponding to $P$-evidences in $S$ in the following manner.
Initially, $T$ is empty and no entry of $S$ is marked. We repeat the following procedure as long as $S$ contains a $P$-evidence whose entries are unamrked:
in each round, we pick the smallest $0 < j \leq n-k+1$ for which both $S_l$ is unmarked for any $j \leq l \leq j+k-2$ and the set $A_j = \{i < j: \text{$S_i$ is unmarked and equal to $1$}\}$ is not empty. Pick $i = \max{A_j}$. Now add $(i,j)$ to $T$ and mark the entries of $S$ in the $k$ locations $i, j, j+1, \ldots, j+k-2$.
Consider the set $T = \{(i_1, j_1), \ldots, (i_t, j_t)\}$ in the end of the process, where $(i_s, j_s)$ was chosen in round $s$ of the process, and let $m$ be the largest integer such that $S_m$ equals $1$ and is not marked (if no such entry exists, we take $m = n$).
Then $T$ satisfies the following:
\begin{itemize}
\item $j_1 < \ldots < j_t$. This is true by the minimality of $j$ picked in any round. In particular, if $S_i$ is marked before $S_j$ and both entries are zero then $i < j$.
\item Let $0 < s \leq t$. If $i_{s} < i_{s-1}$ then we must have $j_{s} = j_{s-1} + k - 1$: by the maximality of $i_s$, it must be true that $P_l = 0$ for any $j_{s-1} + k - 1 \leq l < j_{s}$. Thus, by the minimality of $j_{s}$, we must have that $j_{s} = j_s + k -1$, as desired.
\item On the other hand, if $i_{s} > i_{s-1}$ then $j_{s} = i_{s} + 1$. Again, this is true by the minimality of $j_{s}$ and the maximality of $i_{s}$ with respect to $j_{s}$.
\item For any $1 \leq s \leq t$, either $j_s < m$ or $i_s > m$. This is true by the maximality of $i_s$ with respect to $j_s$.
\end{itemize}
We now show how one can modify exactly $t$ bits in $S$ to make it $P$-free. For any $(i,j) \in T$ such that $j < m$, we modify $P_i$ to zero.
After the modification, $P_l = 0$ for any $ l < m$, so there cannot be $P$-copies in $S$ starting at $l < m$.

On the other hand, for any $(i,j) \in T$ such that $ i > m$, we modify $P_{j+k-2}$ to one. It remains to show that after the modification, $S$ will not contain a $P$-copy starting at some $l \geq m$. By the third observation and the minimality of $j$ in any round of the process, we cannot have a $P$-copy starting in an unmarked entry $S_l = 1$ - since otherwise, at some point in the above process, we would have to add the evidence $(l, l+1)$ to $T$, while marking $S_l$, a contradiction. To complete the proof, suppose that $S_l$ is marked and equals one after the modification; that is, $l = j_s+k-2$ for some $(i_s,j_s) \in T$ with $i_s > m$. There are two cases to consider.
\begin{itemize}
\item If $i_{s+1} < i_s$ then, by the second observation above, $j_{s+1} = j_s + k - 1$. Therefore, the entry in location $j_{s+1} + k - 2 = l + k - 1$ was flipped to one during the modification, so it is not possible to have a $P$-copy starting in location $l$ after the modification.
\item On the other hand, if $i_{s+1} > i_s$ then it is not possible to have $P_{l+1} = P_{l+2} = \ldots = P_{l+k-1} = 0$ before the modification,
since otherwise we would choose the pair $(i'_{s+1}, l+1)$ satisfying $i'_{s+1} \geq m$ in round $s+1$ of the process instead of choosing
$(i_{s+1}, j_{s+1})$, as we have $l+1 < i_{s+1} < j_{s+1}$ and by the first observation above.
\end{itemize}
We conclude that after the modification, $S$ is $P$-free. Hence $d_P(S) \leq |T|$. Since obviously $|T| \leq d_E(S)$, we get that $d_P(S) = d_E(S) = |T|$, as desired.
\end{proof}

\begin{proof}[Proof of Theorem \ref{thm:compute_almost_homo}]
Consider the following algorithm to compute the distance of a string $S$ of length $n$ to $P$-freeness, where $P = 10^{k-1}$.
The idea of the proof is to read $S$ from left to right, marking $P$-evidences along the way and keeping track of the following quantities.
\begin{itemize}
\item The number of $1$'s that were already observed but not yet marked, denoted by $a$.
\item The length of the longest streak of unmarked $0$'s lying immediately to our left, denoted by $b$.
\item The number of non-overlapping $P$-evidences found until now, denoted by $c$.
\end{itemize}
We initialize $a,b,c = 0$. For $i$ running from $0$ to $n-1$ we do as follows:
\begin{itemize}
\item If $S_i = 1$ then $a$ is increased by one.
\item If $S_i = 0$ then
\begin{itemize}
\item If $a = 0$ then we do nothing.
\item if $a > 0$ then $b$ is increased by one. If we now have $b = k-1$ then a new $P$-evidence was found - and we update $a \leftarrow a-1$, $b \leftarrow 0$, $c \leftarrow c+1$.
\end{itemize}
\end{itemize}
The algorithm returns the value of $c$ after the above loop stops.
It is not hard to see that the algorithm indeed calculates $d_E(S) = d_P(S)$. Each round of the loop takes $O(1)$ times, so the total running time is $O(n)$ (with no dependence on $k$). The space complexity is clearly $O(1)$, as we only need to keep track of $a,b,c$.
\end{proof}
%
\begin{proof}[Proof of Theorem \ref{thm:test_almost_homo}]
Let $S$ be a string of length $n$ and suppose that $P = 10^{k-1}$. Given $\epsilon > \delta > 0$ where $\delta \leq \epsilon / (16 + c)$, we will use the following simple tester to distinguish between the case that $S$ is $\epsilon$-far and the case that $S$ is $\delta$-close to $P$-freeness.
Since any string $S$ must be $1/k$-close to $P$-freeness, we may assume that $\epsilon \leq 1/k \leq 1/2$.
The tester makes at most $m$ queries, where $m$ will be determined later.
First it chooses uniformly at random and independently (allowing repetitions) $i_1, \ldots, i_{m/3}, j_1, \ldots, j_{m/3k}$.
Then, the tester samples single entries $S_{i_l}$ for any $1 \leq l \leq m/3$ and blocks $S_{j_r}, \ldots, S_{j_r+2k-2}$ for any $1 \leq r \leq {m/3k}$.
The tester accepts the input if and only if there are no $i_l < j_r$ such that $S_{i_l} = 1$ and $S_{j} = \ldots = S_{j + k-1} = 0$ for some $j_r \leq j \leq j_{r} + k - 1$.

Let $T = \{(x_1, y_1), \ldots, (x_t, y_t)\}$ be a maximal set of non-overlapping $P$-evidences in $S$.
Suppose first that $S$ is $\epsilon$-far from $P$-freeness; then $t \geq \epsilon n$. The probability that none of the locations $x_1, \ldots, x_{t/2}$ was sampled by the tester as a single entry is at most $(1-\epsilon/2)^{m/3} \leq e^{-\epsilon m / 6}$. On the other hand, the probability that none of the $0^{k-1}$-copies starting at locations $y_{t/2 + 1}, \ldots, y_t$ is fully contained in one of the sampled blocks is at most $(1 - \epsilon k / 2)^{m/3k} \leq e^{-\epsilon m / 6}$. In total, the acceptance probability in this case, i.e.\@ the probability that no $P$-evidence was found is at most $2e^{-\epsilon m / 6}$.

On the other hand, if $S$ is $\delta$-close to $P$-freeness then $t \leq \delta n$. Since $T$ is a maximal set of non-overlapping $P$-evidences, to find a $P$-evidence, the tester must either sample at least one of $x_1, \ldots, x_t$ as a single entry or one of the entries of the form $y_i + \delta$, with $1 \leq i \leq t$ and $0 \leq \delta k-1$ as part of one of the sampled blocks.
The probability that the first event does not occur is at least $(1-\delta)^{m/3} > e^{-2 \delta m / 3}$. The probability that the second event does not occur is at least $(1-3k\delta)^{m/3k} > e^{-2\delta m}$. Hence the acceptance probability here, which is the probability that both events do not occur is at least $e^{- 8 \delta m / 3} \geq e^{- \epsilon m / (6 + c/8)}$. Clearly, if we pick $m = c' / \epsilon$ for a large enough constant $c'$ that depends only on $c$, we will get that the quantity $e^{- \epsilon m / (6 + c/8)} - 2e^{-\epsilon m / 6}$ is positive and bounded away from zero; Therefore, with this choice of $m$ the tester accepts $\delta$-close inputs with probability that is bigger by an absolute constant than the acceptance probability for any $\epsilon$-far input. This finishes the proof of the theorem, since one can build a tester that has the same query complexity (up to a multiplicative constant), in which $\delta$-close inputs are accepted with probability at least $2/3$ while $\epsilon$-far inputs are rejected with probability at least $2/3$, as desired.
%
%
%
%
\end{proof}


\section{A Lower Bound for Testing Pattern Freeness}\label{sec.lower.bound}
Here we show that any tester that makes $o_{\epsilon}(\frac{1}{\epsilon})$ queries in an attempt to test if a fixed pattern appears in a string, must err with probability greater than $1/3$. While the lower bound of \cite{Alon} already entails a lower bound of $\Omega(1/\epsilon)$ for testing pattern-freeness (as for any fixed pattern $J$ the language consisting of all strings not containing $J$ is regular), we give here a self-contained proof for two reasons. First, our proof extends to the case where $k$, the length of $J$, is allowed to depend on $n$. Second, it is not hard to adapt the proof to the 2D (or higher dimensions) and demonstrate a lower bound of $\Omega(1/\epsilon)$ on the query complexity of testing pattern freeness of multidimensional arrays.

\begin{theorem}\label{thm:LB}
\textnormal{\textbf{[Lower bound of $\Omega(1/\epsilon)$ queries]}}
Suppose $k$ is even and consider the pattern $J:=0^{k/2-1}10^{k/2}$. Any tester that distinguishes with probability at least $2/3$ between the case that a string $I$ is $J$-free and the case that $I$ is $\epsilon$-far from being $J$-free, makes at least $\frac{1}{13 \epsilon}$ queries to $I$.
\end{theorem}
Note: we show that the lower bound above applies for wide ranges of possible values for $\epsilon$: ranging from $\epsilon=O_n(1/n)$ to $\epsilon=\Omega_k(1/k)$. We did not attempt to optimize the constant $\frac{1}{13}$ in the theorem above. Furthermore, to ease readability we avoid using floor/ceiling signs.
\begin{proof}
By Yao's Principle \cite{Yao}, it suffices to construct two distributions $B$ and $C$ over length $n$ strings where all strings in $B$ are $J$-free and all strings in $C$ are $\epsilon$-far from being $J$-free, with the property that any \emph{deterministic} algorithm $D$ making less than $\frac{1}{13\epsilon}$ queries to $I$ sampled from $A=\frac{1}{2}B+\frac{1}{2}C$ (namely we sample from $B$ with probability $1/2$ and from $C$ otherwise), must err with probability greater than $1/3$. We now explain how to construct these distributions.

The distribution $B$ is just the single string $0^n$ sampled with probability $1$. For the distribution $C$ we take the string $I=0^n$ and divide it to $n/k$ disjoint intervals of length $k$: $I_1,\ldots I_{n/k}$.
Then, we sample randomly a subset of $2\epsilon n$ of these intervals. For each length $k$ interval in the subset, we choose one of the last (right side) $k/2$ locations and flip it to be a 1.

Clearly any string from $B$ is $J$-free with probability $1$, while every string in $C$ contains exactly $2\epsilon n$ occurrences of $J$. Furthermore, as changing a single location in a string $I$ sampled from $C$ can remove at most $2$ occurrences of $J$, such a string $I$ is $\epsilon$-far from being $J$-free. 

Consider a deterministic tester $D$ that performs at most $\frac{1}{13\epsilon}$ queries to the string $I$. Since $D$ is deterministic, in the case that it encounters only zeros during its queries, it has to either declare the string $I$ as being $J$-free or not. 
The first option is that $D$ rejects (declares '$\epsilon$-far from $J$-free') when it sees only zeros. In this case $D$ errs with probability $1$ if $I$ was chosen from $B$, hence it errs with overall probability of at least $1/2$.
We next handle the second case, where $D$ accepts (declares 'free') when it sees only zeros. We show that in this case $D$ errs with probability greater than 1/3, by focusing on the distribution C (i.e. the $\epsilon$-far case).

Being deterministic, when $D$ sees only zeros it queries a \emph{fixed} set $X$ of locations, where we chose $|X|=\frac{1}{13\epsilon}$. All we need to show is that under the random process of the distribution $C$, the fixed set of locations $X$ does not contain a 1 w.p. $> 2/3$. This will suffice, since all strings from $C$ (which is chosen w.p. 1/2) are $\epsilon$-far from being $J$-free, and therefore $D$ will fail w.p. $> 1/2\cdot2/3=1/3$.

Let $X_i=X\cap I_i$ be the set of  locations from $X$ that are in interval $I_i$.
An interval $I_i$ has $k/2$ locations where a 1 could be placed (the right half of the interval) and we can assume that the set $X_i$ is contained in the right half of $I_i$ (since the left half is zero w.p. 1 hence would be pointless to query by any algorithm).
Clearly,
$
Pr[X_i \text{ contains a 1}] = Pr[\text{$I_i$ contains a 1}]\cdot Pr[\text{the 1 is within $X_i$}] =
\frac{2\epsilon n}{n/k}\cdot \frac{|X_i|}{k/2}=4\epsilon|X_i|
$,
 leading to
$
Pr[X \text{ contains a 1}] \le \sum_{i=1}^{n/k} 4\epsilon|X_i| = 4\epsilon|X| < 1/3
$, by union bound.
\end{proof}
As noted, Theorem~\ref{thm:LB} can be generalized to the higher-dimensional settings (details omitted).

%
%




\begin{thebibliography}{11}

\bibitem{Bipartite}
N. Alon (2002).
\newblock Testing subgraphs in large graphs,
\newblock {\em Random structures and algorithms}, 21(3–4):359--370.
\bibitem{ABen}
N. Alon and O. Ben Eliezer (2016).
\newblock{\em Removal lemmas for matrices,}
\newblock arXiv preprint 1609.04235.

\bibitem{Alon}
N. Alon, M. Krivelevich, I. Newman and M. Szegedy (2001).
\newblock Regular languages are testable with a constant number of queries,
\newblock {\em SIAM Journal on Computing}, 30, 1842--1862.

\bibitem{Large}
 N. Alon, E. Fischer, M. Krivelevich and M. Szegedy (2000).
 \newblock Efficient testing of large graphs,
 \newblock {\em Combinatorica}, 20, 451--476.


\bibitem{Induced}
 N. Alon, E. Fischer and I. Newman (2007).
 \newblock Efficient testing of bipartite graphs for forbidden induced subgraphs,
 \newblock {\em SIAM Journal on Computing}, 37.3, 959--976.


\bibitem{AS}
N.~Alon and J.~Spencer (2008).
\newblock {\em The Probablistic Method}.
\newblock Wiley.


\bibitem{Amir0}
A. Amir, G. Benson (1998).
\newblock Two-Dimensional Periodicity in Rectangular Arrays,
\newblock {\em SIAM Journal on Computing}, 27, 90-106.

\bibitem{Amir1}
A. Amir, G. Benson, M. Farach (1994).
\newblock An Alphabet Independent Approach to Two-Dimensional Pattern Matching,
\newblock  {\em SIAM Journal on Computing}, 23, 313-323.

\bibitem{Awasthi}
P. Awasthi, M. Jha, M. Molinaro and S. Raskhodnikova (2016).
\newblock Testing Lipschitz functions on hypergrid domains.
\newblock {\em Algorithmica}, 74(3), 1055--1081.

\bibitem{Bar}
Y. Bar-Hillel, M. Perles, and E. Shamir (1964).
\newblock On formal properties of simple phrase structure grammars,
\newblock In Y. Bar-Hillel, editor, {\em Language and Information: Selected Essays on Their Theory and Application}, 116--150. Addison-Wesley, Reading, Massachusetts.

\bibitem{Berman}
P. Berman, M. Murzabulatov, S. Raskhodnikova (2015).
\newblock Constant-Time Testing and Learning of Image Properties,
\newblock arXiv prepreint 1503.01363.

\bibitem{Berman2}
P. Berman, M. Murzabulatov and Sofya Raskhodnikova (2016).
\newblock Tolerant Testers of Image Properties.
\newblock {\em ICALP}, 1--90:14.

\bibitem{Boyer}
R.S. Boyer and J.S. Moore (1977).
\newblock A fast string searching algorithm,
\newblock {\em Comm. ACM}, 20(10), 762--772.

\bibitem{Cole}
R. Cole (1991).
\newblock Tight Bounds on the Complexity of the Boyer-Moore String Matching Algorithm,
\newblock {\em SODA}, 224--233.

\bibitem{Czumaj}
C. Maxime, A. Czumaj, L. Gasieniec, S. Jarominek, T. Lecroq, W. Plandowski, and W. Rytter (1994).
\newblock Speeding up two string-matching algorithms,
\newblock {\em Algorithmica}, 12, 247--267.

\bibitem{NewmanStoc}
E. Fischer, and I. Newman (2001).
\newblock Testing of matrix properties,
\newblock {\em STOC}, 286--295

\bibitem{Newman}
E. Fischer and I. Newman (2007).
\newblock Testing of matrix-poset properties,
\newblock {\em Combinatorica}, 27(3), 293--327.


\bibitem{Rozenberg}
E. Fischer, E. Rozenberg,
\newblock Lower bounds for testing forbidden induced substructures in
bipartite-graph-like combinatorial objects,
\newblock  Proc. RANDOM 2007, 464-478.

\bibitem{Fowler}
R. Fowler, M. S. Paterson and S. L. Tanimoto (1981).
\newblock Optimal packing and covering in the plane are NP-complete.
\newblock {\em Information processing letters}, 12(3), 133--137.

\bibitem{Galil3D}
Z. Galil, J. G. Park and K. Park. (2004).
\newblock Three-dimensional periodicity and its application to pattern matching.
\newblock {\em SIAM Journal on Discrete Mathematics}, 18(2), 362--381.

\bibitem{Galil}
Z. Galil and J. I. Seiferas (1983).
\newblock Time-Space-Optimal String Matching,
\newblock {\em J. Comput. Syst. Sci}, 26(3), 280--294.

\bibitem{Goldreich}
O. Goldreich, S. Goldwasser and D. Ron (1998).
\newblock Property testing and its connection to learning and approximation,
\newblock {\em JACM}, 45, 653--750.

\bibitem{Kark}
J. Kärkkäinen and E. Ukkonen (2008).
\newblock Multidimensional string matching.
\newblock In {\em Encyclopedia of Algorithms}, 559--562.

\bibitem{Kleiner}
I. Kleiner, D. Keren, I. Newman, O. Ben-Zwi (2011).
\newblock Applying Property Testing to an Image Partitioning Problem,
\newblock {\em IEEE Trans. Pattern Anal. Mach. Intell}, 33(2), 256--265.

\bibitem{Korman}
S. Korman, D. Reichman, G. Tsur and S. Avidan.
\newblock Fast-Match: Fast Affine Template Matching,
\newblock {International Journal of Computer Vision}, to appear.

\bibitem{knuth}
D. E. Knuth, J. H. Morris Jr. and V. R. Pratt (1977).
\newblock Fast Pattern Matching in Strings,
\newblock {\em SIAM J. Comput}. 6(2): 323--350.

\bibitem{Lecroq}
T. Lecroq  (2007).
\newblock Fast exact string matching algorithms,
\newblock {\em Information Processing Letters}, 102(6), 229-235.

\bibitem{Navarro}
G. Navarro and M. Raffinot (2000).
\newblock Fast and flexible string matching by combining bit-parallelism and suffix automata,
\newblock {\em Journal of Experimental Algorithmics (JEA)} 5: 4.

\bibitem{Parnas}
M. Parnas, D. Ron and R. Rubinfeld (2006).
\newblock Tolerant property testing and distance approximation.
\newblock {\em Journal of Computer and System Sciences}, 72(6), 1012--1042.

\bibitem{sofya}
S. Raskhodnikova (2003).
\newblock Approximate testing of visual properties,
\newblock{\em RANDOM}, 370--381.

\bibitem{Rivest}
R. L. Rivest (1977).
\newblock On the Worst-Case Behavior of String-Searching Algorithms,
\newblock {\em SIAM J. Comput.} 6(4): 669--674.

\bibitem{Rubinfeld}
R. Rubinfeld and M. Sudan (1996).
\newblock Robust characterization of polynomials with applications to program testing,
{\em SIAM J. Comput.} 25, 252--271.

\bibitem{Ron}
G. Tsur and D. Ron (2014).
\newblock Testing properties of sparse images,
\newblock {\em ACM Transactions on Algorithms} 4.

\bibitem{Yao}
A. C. Yao (1977).
\newblock {Probabilistic computation, towards a unified measure of complexity,}
\newblock {\em FOCS}, 222--227.

\end{thebibliography}


\end{document}